\documentclass[11pt]{article}
\usepackage[a4paper,margin=1in]{geometry}
\usepackage{setspace}
\usepackage[T1]{fontenc}
\usepackage[utf8]{inputenc}
\usepackage{lmodern}
\usepackage{microtype}
\usepackage{amsmath,amssymb,amsthm,amsfonts,array,mathrsfs,ifthen,mathtools}
\usepackage[round,authoryear]{natbib}
\usepackage{dsfont}
\usepackage{enumitem}
\usepackage{comment}
\usepackage{url}
\usepackage{bm}
\usepackage{graphicx}
\usepackage{xcolor}
\usepackage{pgfplots}
\pgfplotsset{compat=1.18}
\usepackage[colorinlistoftodos,textsize=tiny]{todonotes}
\usepackage{hyperref}
\definecolor{warmdarkred}{RGB}{120, 20, 30}
\definecolor{myblue}{RGB}{0,30,180}
\hypersetup{
    colorlinks=true,
    linkcolor=warmdarkred,
    urlcolor=warmdarkred,
    citecolor=warmdarkred
}
\theoremstyle{plain}
\newtheorem{theorem}{Theorem}
\newtheorem{proposition}{Proposition}
\newtheorem{lemma}{Lemma}
\newtheorem{corollary}{Corollary}
\newtheorem{assumption}{Assumption}

\newtheorem{example}{Example}

\newcommand{\E}{\mathbb E}

\newcommand{\ubar}{\underline}
\newcommand{\vbar}{\overline}

\newcommand{\D}{\Delta}

\DeclareMathOperator*{\argmax}{arg\,max}
\DeclareMathOperator*{\argmin}{arg\,min}

\title{Shadow-score auctions for execution incentives
}
\author{Federico Vaccari\\
        {\small University of Bergamo}\\
        {\small \texttt{vaccari.econ@gmail.com}}}
\date{}

\begin{document}
\maketitle

\thispagestyle{empty}

\begin{abstract}
This paper studies optimal auctions in which the allocation creates a moral hazard problem for a third-party executor. The executor chooses effort before the winner is known. In regular independent-private-values environments, the optimal mechanism is a shadow-score auction. Each bidder's virtual value is adjusted by the shadow value of relaxing the executor's incentive constraint using that bidder's allocation state. Unlike standard scoring auctions, the score is derived from a non-bidder's moral hazard constraint rather than from a preference for quality. Reserve formats miss this ranking channel, as they can adjust whether the object is sold, but cannot favor allocation states that are more useful for motivating execution effort. The paper also studies cases in which this scoring representation breaks down and the optimal mechanism becomes a constrained shadow allocation.
\end{abstract}

\bigskip

 \medskip



\newpage

\thispagestyle{empty}

\begingroup
\hypersetup{linkcolor=black}
\tableofcontents
\endgroup

\newpage


\section{Introduction}

Standard auction design frequently focuses on allocation and payments. Once the winner is selected and payments are made, the mechanism has largely done its job. Many economically important auctions do not have this property. Assigning the object creates a task that must still be carried out. A license must be deployed, a concession monitored, an asset approved or restructured, or a platform participant screened and supervised. In these settings, the auction does more than select a winner, as it also chooses the execution problem that must be implemented.

This paper studies optimal auction design when assigning the object creates an execution task that requires non-contractible effort by a third party. I use \emph{execution} broadly, to include the implementation, monitoring, enforcement, approval, or realization of the allocation once the auction outcome is known. The executor may be a regulator, a monitoring agency, an enforcement unit, a platform compliance team, or an administrative authority. Because effort itself cannot be contracted on, incentives must be provided through bonuses paid when execution succeeds. Both the effectiveness and the expected cost of these bonuses depend on who wins the auction. The mechanism does not merely choose whom to allocate to. It also determines the environment in which the execution contract must operate.

The central question is whether this problem leaves the standard optimal auction intact. The answer is no in general. If the executor's incentive depends only on whether the object is sold, the execution constraint can be absorbed into a reserve-price adjustment. If the incentive depends on the identity of the winner, the allocation rule changes. The optimal mechanism is a shadow-score auction, where bidders are ranked by virtual values adjusted for the shadow value of relaxing the executor's incentive constraint.

The model keeps the bidder side close to the canonical independent private values (IPV) environment. There is a single indivisible object, and risk-neutral bidders with independent private values. The seller commits to an auction and to a reward schedule for the execution agent. Before values are realized and before the winner is known, the executor chooses a non-contractible effort level. Effort increases the probability of a verifiable success, but its effect and the pledgeable reward available upon success may differ among potential winners. The seller cares about auction revenue and execution surplus, net of expected execution rewards and other winner-specific costs.

The timing and nature of execution capacity are central. The baseline model captures settings in which the executor must prepare before the winner is known. A regulator may need to allocate monitoring resources, an agency may need to reserve enforcement capacity, a platform may need to prepare compliance infrastructure, or an administrative authority may need to assemble a review team. The executor makes a single general preparation effort that will be used regardless of which allocation is eventually realized. The winner still matters because the same execution capacity may be more productive, more pledgeable, or less costly under some winners than under others.

The main result characterizes the optimal auction mechanism. The analysis rewrites the seller's objective in terms of the allocation rule and promised execution rewards. The multiplier on the executor's incentive constraint then generates a bidder-specific shadow adjustment to virtual values. The optimal auction ranks bidders by their adjusted virtual values and allocates the object to the bidder with the highest nonnegative shadow score. Under regularity, the resulting allocation is monotone in each bidder's value and is implementable in dominant strategies.

This characterization clarifies the limits of familiar auction formats. In a symmetric value environment, a common-reserve auction works when execution does not create bidder-specific advantages in the allocation score. Common execution terms are one way this can happen, but they are not the only way. When execution adjustments differ across bidders, the optimal auction generally requires bidder-specific scores. Reserve prices can decide whether the object is sold. Shadow scores can also affect who among the eligible bidders wins. This ranking channel is the distinctive feature of the model, as part of the score is generated endogenously by the executor's incentive constraint.

The ranking channel is important for policy auctions and procurement. In highway procurement, for example, cost-plus-time, or $A+B$, bidding combines a contractor's price bid with its promised completion time, valued at a road-user cost per day \citep{FHWA_AB_Bidding}. More generally, best-value procurement evaluates not only price but also the bidder's ability to perform the contract successfully, including past performance and execution-related criteria \citep{FAR_15305,FAR_421501}. These practices reflect a basic economic concern. The highest monetary bid need not be the bid that creates the best execution problem. The model provides a mechanism-design foundation for such adjustments when execution requires incentives for a third-party executor.

The analysis delivers four main implications. First, the score is partly endogenous. Some bidder-specific adjustments, such as direct execution payoffs, are ordinary payoff shifters. The new feature is that the auction also creates a score adjustment through the executor's incentive constraint. This shadow adjustment depends on how useful each bidder's allocation state is for motivating execution effort, and must be determined jointly with the execution reward schedule.

Second, the auction and the execution contract must be designed jointly. Execution rewards are useful only in the allocation states in which the corresponding bidder wins, while the value of those rewards feeds back into the allocation score. Third, reserve formats generally miss the ranking channel. Common reserves change the sale threshold, and bidder-specific reserves change eligibility, but shadow scores can change pairwise rankings among eligible bidders.

Fourth, the public-primitives benchmark identifies when the optimal allocation has a reduced-form scoring representation. The later results show why that representation can fail. Private reliability limits what can be elicited. Certification determines what can be scored. Execution externalities can destroy monotonicity. Set-dependent execution replaces individual scores with set-level shadow values.

Scoring auctions are well understood in procurement and multidimensional auction design. This paper derives bidder-specific score adjustments from a moral hazard constraint faced by a non-bidder. Conditional on the optimal multiplier and public winner-specific execution primitives, the allocation is observationally equivalent to a scoring auction with bidder-specific payoff shifters. Those shifters are endogenous shadow values generated by the executor's incentive constraint and by the feasible success-contingent reward schedule. The equivalence becomes incomplete once the execution information or technology becomes richer. Private reliability, certification, report-dependent execution values, and set-dependent execution all change what can be represented by individual quality scores. The analysis provides both a contracting foundation for scoring rules and a boundary for the standard scoring interpretation.

The analysis proceeds as follows. Section~\ref{sec:model} introduces the environment. Section~\ref{sec:analysis} characterizes the optimal shadow-score auction when execution primitives are public and winner-specific, and compares ex ante execution capacity with no effort and post-win effort. Section~\ref{sec:format_restrictions} studies reserve-based formats and quantifies the loss from common reserves in a two-bidder example. Section~\ref{sec:beyond_scoring} studies environments in which the reduced-form scoring representation is no longer sufficient. Finally, Section~\ref{sec:conclusion} concludes.


\subsection{Related literature}

This paper relates first to the theory of optimal auctions. \citet{Myerson1981} characterizes revenue-maximizing mechanisms in single-object private-value environments through virtual values, and \citet{RileySamuelson1981} gives a closely related reserve price characterization. I keep the bidder side of the model within this canonical single-dimensional environment. The main departure is that the allocation also creates a moral hazard problem for a third-party executor. This adds an execution adjustment to the usual virtual-value ranking. Part of this adjustment may simply reflect direct winner-specific execution payoffs. The new feature is that the executor's incentive constraint itself generates another part. When this shadow component differs across bidders, the auction may need to adjust bidder rankings to make execution effort incentive compatible.\footnote{The standard Myerson allocation is recovered as a special case when execution does not affect either bidder rankings or the sale threshold.}

The paper also contributes to the literature on scoring and multidimensional auctions. \citet{Che1993} studies design competition through multidimensional auctions, \citet{Branco1997} studies multidimensional procurement mechanisms, and \citet{AskerCantillon2008,AskerCantillon2010} analyze scoring auctions and procurement when price and quality matter. In that literature, the score aggregates contractible bid attributes according to the buyer's preferences or an optimally chosen procurement rule. Here, the non-price component of the score has a different source, as it is the shadow value of relaxing the executor's incentive constraint rather than a quality weight.  With public execution parameters, the optimal allocation is observationally equivalent to a scoring auction with endogenous bidder-specific payoff shifters. Section~\ref{sec:beyond_scoring} shows where this reduced-form scoring interpretation breaks down.\footnote{A related use of bidder-specific handicaps appears in \citet{EsoSzentes2007}, who show that an optimal information-disclosure mechanism can be implemented by a handicap auction in some applications. In the present paper, the handicap is not chosen as part of an information-disclosure mechanism. It is the shadow value of relaxing the executor's incentive constraint.}

The paper is also related to auctions combined with incentive contracts. \citet{McAfeeMcMillan1986} study bidding for contracts in the presence of moral hazard, and \citet{LaffontTirole1987} show how auction theory and incentive theory can be combined when the selected firm later exerts cost-reduction effort. \citet{LaffontTirole1993} develop the broader theory of incentives in procurement and regulation. More recently, \citet{ChakrabortyKhalilLawarree2021} study competitive procurement with ex post moral hazard, and show that optimal procurement may require limiting competition or using inefficient allocation rules. In these papers, the incentive problem is tied to the winning contractor's post-award behavior. In contrast, I study an executor whose effort is chosen before the winner is known, and who is not one of the bidders.

At the contracting level, the paper builds on moral hazard models with limited liability. \citet{Holmstrom1979} studies moral hazard under imperfect observability, and \citet{Innes1990} shows how limited liability shapes optimal incentive contracts. In the present model, the incentive contract is embedded in an auction-design problem. Limited liability restricts the executor's bonus, and the value of relaxing that restriction depends on the winning allocation state. Thus, different winners generate different marginal returns to success-contingent rewards, which feeds back into the auction allocation rule.

The analysis is also related to auctions with externalities. \citet{JehielMoldovanuStacchetti1996,JehielMoldovanuStacchetti1999} show that winner identity and payoff externalities can substantially change auction design. In the present model, winner identity matters through an incentive channel rather than through direct payoff externalities among bidders. Assigning the object to one bidder rather than another changes the productivity, pledgeability, or expected cost of the executor's reward. The externality is one of execution, as the auction outcome changes the cost of inducing effort on the outside of the bidder side of the market.

The paper is more distantly related to auctions with ex ante investments or endogenous valuations, such as \citet{ArozamenaCantillon2004} and \citet{GershkovMoldovanuStrackZhang2021}. There, agents take costly actions that affect their own costs or values before the auction. Here, the costly action is chosen by a non-bidder executor, and it affects the allocation's implementability.

Finally, the private reliability results distinguish the paper from standard scoring models in another way. If reliability is public or verifiable, it can enter the shadow score directly. If reliability is private and payoff-irrelevant to the bidder, however, it cannot generally be elicited by a direct mechanism without verification, costly signaling, or some other source of payoff relevance. The designer cannot simply append an unverifiable execution characteristic to a one-dimensional auction score.


\section{Environment}
\label{sec:model}

There is a seller with one indivisible object, and a set $N=\{1,\ldots,n\}$ of bidders with unit demand. Bidder $i$ has private value $v_i\in[\ubar v_i,\vbar v_i]$, independently drawn from distribution $F_i$ with density $f_i>0$ on its support. Bidders are risk neutral and have quasi-linear utility. The seller has no value for retaining the object. A direct mechanism consists of allocation rules $x_i(v)\in[0,1]$ and bidder payments $p_i(v)$, where $v=(v_1,\ldots,v_n)$ and $\sum_{i=1}^n x_i(v)\leq 1$ for every $v$. If bidder $i$ receives the object and pays $p_i$, its payoff is $v_i-p_i$. 

After the allocation, the assigned object must be \emph{executed}. Execution means the realization of the allocation, which can either succeed or fail. A separate agent, the executor, chooses effort $e\in\{0,1\}$ before the realization of bidder values and before the winner is known.\footnote{Appendix~\ref{subsec:continuous_effort} extends the binary effort baseline to a continuous effort choice.} Effort costs $\psi>0$, and is not observable. If bidder $i$ receives the object, the probability of execution success is
\[
\sigma_i^0+e\D_i,
\]
where $\D_i\geq0$ and $0\leq \sigma_i^0\leq 1$. Moreover, for every $i\in N$, 
\[
\sigma_i\coloneqq \sigma_i^0+\D_i\in[0,1]
\]
The parameter $\D_i$ is the marginal effect of execution effort when bidder $i$ is the winner. It captures how useful the execution agent's effort is under that allocation. The parameter $\sigma_i$ is the success probability under effort.

The execution agent is risk neutral and has outside option normalized to zero. Limited liability requires all transfers to the agent to be nonnegative, so the seller can reward success but cannot impose penalties after failure. Success is verifiable, so success-contingent bonuses are contractible, but effort itself is not.

The seller can promise the execution agent a success-contingent bonus after the auction outcome. If bidder $i$ is assigned the object at report profile $v$, let $b_i(v)\geq0$ denote the bonus paid to the execution agent if execution succeeds. The bonus is relevant only on histories in which bidder $i$ receives the object. Limited liability and budget constraints impose $0\leq b_i(v)\leq B_i$ whenever $x_i(v)>0$, where $B_i\geq0$ is the maximum pledgeable success-contingent bonus after assigning the object to bidder $i$.

It is convenient to work with promised bonus expenditures, defined as
\[
q_i(v)\coloneqq x_i(v)b_i(v).
\]
It follows that the feasible bonus-expenditure variables satisfy $0\leq q_i(v)\leq B_i x_i(v)$ for every $i$ and $v$. Whenever $x_i(v)>0$, the corresponding bonus is $b_i(v)=q_i(v)/x_i(v)$. When $x_i(v)=0$, the bonus is payoff-irrelevant and $q_i(v)=0$. This formulation makes the relaxed problem linear in $(x,q)$. The restriction to success-contingent bonuses is without loss under the contracting environment considered here.\footnote{Suppose that after assigning the object to bidder $i$, the agent can be paid $t_i^S$ after success and $t_i^F$ after failure. The incentive constraint depends only on the payment spread, i.e., on $\Delta_i(t_i^S-t_i^F)\geq \psi$. Since payments are subject to limited liability, $t_i^S,t_i^F\geq0$, the cost-minimizing way to create any given incentive spread sets $t_i^F=0$. Fixed transfers that do not vary with success do not affect effort incentives, and are dominated when the agent's outside option is normalized to zero. As a result, the contract can be represented by a success-contingent bonus $b_i(v)$ without loss.}

The execution agent exerts effort if its expected incremental bonus from effort is at least $\psi$. An allocation-bonus pair $(x,q)$ induces effort if
\begin{equation*}
\E\left[\sum_{i=1}^n \D_i q_i(v)\right]\geq \psi,
\end{equation*}
where the expectation is taken over bidder values.

If bidder $i$ receives the object and effort is exerted, the seller obtains an execution payoff $g_i\in\mathbb{R}$. This term captures the expected value of successful execution under effort net of execution costs other than the bonus. Since success occurs with probability $\sigma_i$ under effort, the expected cost of an effective success-contingent bonus is $\sigma_i q_i(v)$. Conditional on inducing effort, the seller's expected payoff is
\begin{equation*}
\E\left[ \sum_{i=1}^n p_i(v) + \sum_{i=1}^n x_i(v)g_i - \sum_{i=1}^n \sigma_i q_i(v) \right].
\end{equation*}
The seller may choose not to induce effort. The main characterization below is conditional on inducing effort. The seller induces effort if the resulting payoff exceeds the value of the best no-effort mechanism, which is described after the main theorem.

\begin{assumption}
\label{ass:regularity}
For every bidder $i$, the virtual value
\[
\varphi_i(v_i) \coloneqq v_i-\frac{1-F_i(v_i)}{f_i(v_i)}
\]
is strictly increasing.
\end{assumption}

The timing is as follows. First, the seller commits to an auction mechanism and to a success-contingent reward schedule for the execution agent. Second, before bidder values are realized and before the winner is known, the execution agent chooses a noncontractible effort level. This effort is interpreted as an ex ante execution-capacity investment.\footnote{For example, preparing monitoring resources, allocating enforcement capacity, assembling a review team, coordinating complementary agencies, or building compliance infrastructure.} Third, bidder values are realized, the auction is run, and the object is assigned. Finally, execution succeeds or fails.

Two assumptions are important. First, the effort is chosen before the winner is known. Second, effort is pooled. That is, the executor chooses a single overall preparation effort rather than separate efforts for different bidders or execution paths. The effort may be more productive under some winners than under others, but the executor chooses only a single capacity level. This is why the moral hazard constraint is aggregate. If effort were chosen after the winners were known, then the constraint would be winner-by-winner. If ex ante preparation could be directed separately toward different bidders, the problem would involve several distinct effort choices rather than the single pooled effort studied here.



\section{Observable execution and shadow scores}\label{sec:analysis}


\subsection{Fixed execution rewards}
\label{sec:fixed}

Before studying the joint design of bonuses and auctions, it is useful to analyze the simpler case in which the execution reward is fixed. Suppose that the executor receives an exogenous success bonus $B>0$ whenever the object is sold and execution succeeds. The effort constraint becomes
\begin{equation}
B\,\E\left[\sum_{i=1}^n \D_i x_i(v)\right]\geq\psi.
\label{eq:fixedIC}
\end{equation}
The bonus cost under effort is $B\sum_i \sigma_i x_i(v)$. 

\begin{lemma}
\label{prop:fixed}
Suppose that the execution reward is fixed at $B$. For any optimal multiplier $\lambda\geq0$ on the effort constraint, the optimal allocation assigns the object to a bidder with the highest nonnegative score $S_i$, where
\begin{equation}
S_i(v_i\mid \lambda)=\varphi_i(v_i)+g_i-\sigma_i B+\lambda B\D_i.
\label{eq:fixedscore}
\end{equation}
Under Assumption~\ref{ass:regularity}, this allocation is implementable by a dominant-strategy scoring auction with critical value payments.
\end{lemma}

Lemma~\ref{prop:fixed} identifies the basic force. Suppose that $\Delta_i=\Delta$ for every bidder $i$, and suppose $g_i-\sigma_iB= h$ for every bidder $i$. Then, the execution constraint depends only on the probability of sale, and the execution payoff term is common across bidders. In a symmetric regular environment, the optimal mechanism remains a second-price auction with a reserve, although the reserve may change. If $g_i-\sigma_iB+\lambda B\Delta_i$ differs across bidders, the allocation compares virtual values plus bidder-specific score shifts. Heterogeneity in $g_i-\sigma_iB$ affects rankings directly. Heterogeneity in $\Delta_i$ affects rankings through the incentive term only when the effort constraint has positive shadow value.


\subsection{Endogenous execution contracts}
\label{sec:endogenous}

We now allow the seller to choose the execution contract jointly with the auction. Working with promised bonus expenditure $q_i(v)=x_i(v)b_i(v)$, the effort-inducing problem is linear in the allocation and bonus-expenditure variables. The execution constraint is now
\[
\E\left[\sum_{i=1}^n \D_i q_i(v)\right]\geq \psi,
\]
and feasible bonus expenditures satisfy $0\leq q_i(v)\leq B_i x_i(v)$.

For $\lambda\geq0$, define the \emph{execution value} of assigning the object to bidder $i$ by
\begin{equation}
\chi_i(\lambda) \coloneqq \max_{b\in[0,B_i]}(\lambda\,\D_i-\sigma_i)b = B_i\max\{\lambda\D_i-\sigma_i,0\}.
\label{eq:chi}
\end{equation}
A success bonus for bidder $i$ has positive shadow value only when its incentive benefit, $\lambda\D_i$, exceeds its expected cost, $\sigma_i$. Define the \emph{shadow score}
\begin{equation*}
S_i(v_i\mid \lambda) \coloneqq \varphi_i(v_i)+g_i+\chi_i(\lambda).
\end{equation*}
For any allocation-bonus pair $(x,q)$, write
\[
I(x,q)\coloneqq \E\left[\sum_{i=1}^n \D_i q_i(v)\right]
\]
for the execution incentive generated by the mechanism, and
\[
\Pi(x,q) \coloneqq \E\left[ \sum_{i=1}^n x_i(v)\left(\varphi_i(v_i)+g_i\right) - \sum_{i=1}^n \sigma_i q_i(v) \right]
\]
for the seller's expected virtual payoff, net of execution bonuses.

Moreover, I impose the following strict feasibility condition for inducing effort,
\begin{equation}
\exists\,(\tilde x,\tilde q) \;\text{ such that }\; \sum_{i=1}^n \tilde x_i(v)\leq 1,\;\; 0\leq \tilde q_i(v)\leq B_i\tilde x_i(v), \;\text{ and }\; I(\tilde x,\tilde q)>\psi.
\label{eq:slater}
\end{equation}
This condition rules out the boundary case in which effort is feasible only exactly at the executor's incentive constraint.

The analysis below works with relaxed, possibly randomized mechanisms. Allocation and bonus-expenditure rules are assumed to be measurable, and the relaxed optimization problems are assumed to attain their values. These are standard regularity conditions in mechanism-design problems with bounded primitives and compact type spaces. With these conventions in place, the multiplier argument below identifies the shadow value of the executor's incentive constraint.

Before characterizing the optimal effort-inducing mechanism, it is useful to separate the feasibility of effort from the optimization problem itself. Because the executor's incremental reward from effort is bounded by the largest incentive that can be generated in any allocation state, the model has a simple feasibility threshold. Lemma~\ref{lem:effort_feasibility} identifies this threshold, and shows that the strict-feasibility condition used below is exactly the requirement that the effort cost lie below it.


\begin{lemma}
\label{lem:effort_feasibility}
Let $\bar I\coloneqq \max_{i\in N}\D_i B_i$. Inducing effort is feasible if and only if $\psi\leq \bar I$. It is strictly feasible, in the sense of \eqref{eq:slater}, if and only if $\psi<\bar I$.
\end{lemma}


The boundary case $\psi=\bar I$ is feasible but non-robust. Indeed, effort can be induced only by using allocation states that attain the maximal incentive capacity $\bar I$. I impose strict feasibility in the characterization below not because effort is impossible at the boundary, but because the supporting multiplier argument and the shadow-score representation are cleanest when the executor's incentive constraint has slack feasible directions.

The next step is to explain why a shadow value can summarize the execution constraint. The seller chooses allocations and bonus expenditures state by state, but the executor's incentive constraint is aggregate, as it depends on the incentives generated across all possible allocation states. Under strict feasibility, this aggregate constraint has a supporting multiplier. That multiplier is the shadow value that will enter the allocation rule. The next lemma formalizes this point.

\begin{lemma}
\label{lem:multiplier}
Suppose that~\eqref{eq:slater} holds and that the relaxed effort-inducing problem attains its value. Then, there exists a multiplier $\lambda^\ast \ge 0$ supporting an optimal effort-inducing allocation-bonus pair. Alternatively, an optimal allocation-bonus pair maximizes the Lagrangian associated with $\lambda^\ast$, and complementary slackness holds.
\end{lemma}

The preceding lemmata separate the two ingredients needed for the main characterization. Effort is feasible exactly when the execution requirement is below the maximal incentive capacity, and, under strict feasibility, the executor's incentive constraint admits a supporting multiplier. Theorem~\ref{thm:execution_score} uses this multiplier to characterize the optimal effort-inducing mechanism. For a fixed shadow value, the seller chooses both the winner and the success-contingent bonus state by state. Bonus expenditures are used only when their incentive value exceeds their expected cost, and the maximized net value enters bidder $i$'s score through $\chi_i(\lambda)$. Under regularity, the resulting shadow-score allocation is monotone and is implemented by critical value payments.

\begin{theorem}
\label{thm:execution_score}
Suppose that Assumption~\ref{ass:regularity} and the strict feasibility condition~\eqref{eq:slater} hold. For each $\lambda\geq0$, define $\chi_i(\lambda)$ as in \eqref{eq:chi}. Then, the following statements hold.

\begin{enumerate}[label=(\roman*)]

\item \textbf{Optimal score allocation.}
There exists an optimal effort-inducing mechanism and a multiplier $\lambda^\ast\geq0$ such that the object is assigned only to bidders with maximal nonnegative shadow score. That is,
\[
x_i^\ast(v)>0 \implies S_i(v_i\mid \lambda^\ast) = \max\left\{0,\max_{j\in N}S_j(v_j\mid \lambda^\ast)\right\}.
\]
For every profile $v$ and bidder $i$, the associated bonus expenditure
satisfies
\[
q_i^\ast(v) \in \argmax_{0\leq q\leq B_i x_i^\ast(v)} (\lambda^\ast\D_i-\sigma_i)q.
\]

\item \textbf{Implementation by a shadow-score auction.}
The optimal allocation is dominant-strategy implementable by a shadow-score auction with critical value payments. Fix an exogenous priority order over bidders, used only to break ties among bidders with the same maximal nonnegative shadow score. For each bidder $i$, define
\begin{equation}\label{eq:m}
M_i(v_{-i}\mid \lambda^\ast) \coloneqq \max\left\{ 0,\max_{j\neq i}S_j(v_j\mid \lambda^\ast) \right\}.
\end{equation}
Bidder $i$ wins if and only if $S_i(v_i\mid \lambda^\ast) \geq M_i(v_{-i}\mid \lambda^\ast)$. If more than one bidder is eligible because of a score tie, the priority rule
selects the winner. When bidder $i$ wins, its payment is the critical value
\[
p_i(v_{-i}) = \inf\left\{ z\in[\underline v_i,\overline v_i] \mid S_i(z \mid \lambda^\ast) \geq M_i(v_{-i}\mid \lambda^\ast) \right\},
\]
where the critical value is kept within bidder $i$'s value support.
\end{enumerate}
\end{theorem}


The theorem identifies the optimal allocation and the associated bonus expenditures. The next corollary rewrites the shadow term in a way that makes its economic content more transparent. A bidder's allocation state is useful for providing execution incentives only if the shadow value of the incentive it generates exceeds the expected cost of the success-contingent reward. The ratio $\sigma_i/\Delta_i$ plays the role of an execution-incentive cost index, as it is the expected bonus cost per unit of incentive generated by assigning the object to bidder $i$.


\begin{corollary}
\label{cor:execution_cost_index}
For every bidder with $\D_i>0$, define $\rho_i\coloneqq \sigma_i/\D_i$. For bidders with $\D_i=0$, set $\rho_i=+\infty$. Then, the shadow term can be written as
\[
\chi_i(\lambda) = B_i\D_i\max\{\lambda-\rho_i,0\}.
\]
Thus, bidder $i$'s allocation state receives a positive success-contingent bonus only if $\lambda^\ast\geq \rho_i$. When the inequality is strict, the optimal bonus is maximal. That is, $b_i^\ast(v)=B_i$ whenever $x_i^\ast(v)>0$. When $\lambda^\ast<\rho_i$, bidder $i$'s allocation state is not used to provide execution incentives, and $b_i^\ast(v)=0$.
\end{corollary}


Corollary~\ref{cor:execution_cost_index} gives a simple economic interpretation of the shadow score. The ratio $\rho_i=\sigma_i/\Delta_i$ is the expected bonus cost of creating one unit of execution incentive through bidder $i$'s allocation state. The mechanism first uses allocation states with low $\rho_i$, because they are cheaper ways to motivate the executor. As the shadow value $\lambda^\ast$ rises, states with higher $\rho_i$ become worth using. The score premium is larger for bidders whose allocation states combine high pledgeable rewards, high marginal productivity of effort, and low expected bonus cost. This gives the execution contract a merit-order interpretation.

The theorem nests the standard \cite{Myerson1981} logic. If the executor's incentive constraint is slack, then $\lambda^\ast=0$ and $\chi_i(\lambda^\ast)=0$ for every bidder $i$. In that case, the moral hazard constraint does not affect bidder rankings. If, in addition, the direct execution payoff $g_i$ is common across bidders, the allocation reduces to the usual virtual-value allocation, implemented in the symmetric regular case by a second-price auction with a reserve. More generally, when $\lambda^\ast=0$, any departure from the standard Myerson allocation comes only from direct winner-specific execution payoffs, not from the executor's incentive constraint.

The theorem is stated in terms of bonus expenditures $q_i^\ast(v)$, rather than only winner-level bonuses. When $\lambda^\ast\Delta_i\neq\sigma_i$, the optimal expenditure is pinned down, and it is either zero or the maximal feasible amount $B_i x_i^\ast(v)$. When $\lambda^\ast\Delta_i=\sigma_i$, the seller is locally indifferent over all feasible bonus expenditures for bidder $i$. In that knife-edge case, the optimal expenditure is selected from the primal solution. If $\lambda^\ast>0$, complementary slackness implies
\[
\mathbb E\left[\sum_{i=1}^n \Delta_i q_i^\ast(v)\right]=\psi.
\]
Thus, the characterization does not require a generic marginal bidder. Any fine-tuning of the incentive constraint can occur through bonus choices over states where the seller is locally indifferent.

The term \emph{shadow-score} auction refers to the fact that the allocation score contains the shadow value of the executor's incentive constraint. Unlike in a standard scoring auction, the non-price component is not a quality preference. Rather, it is the value of assigning the object to a bidder because that allocation state helps provide execution incentives. Section~\ref{sec:format_restrictions} shows that this shadow adjustment is not generally equivalent to a reserve-price policy, because it can change pairwise rankings among eligible bidders.

The shadow score depends on the multiplier $\lambda^\ast$, so comparative statics of the allocation pass through the shadow value of the executor's incentive constraint. For a candidate multiplier $\lambda$ and bonus caps $B=(B_1,\ldots,B_n)$, define the maximized shadow-score surplus
\begin{equation}\label{eq:surplus}
H(\lambda\mid B) \coloneqq  \E\left[ \max\left\{ 0, \max_{i\in N} \left( \varphi_i(v_i)+g_i + B_i\max\{\lambda\Delta_i-\sigma_i,0\} \right) \right\} \right].
\end{equation}
This is the value of the shadow-score allocation for a given candidate shadow value $\lambda$, after optimizing over success-contingent bonus expenditures but before subtracting the term $\lambda\psi$ associated with the required execution incentive. Thus, the dual objective is $H(\lambda\mid B)-\lambda\psi$, and the set of optimal multipliers is
\[
\Lambda(\psi\mid B) = \argmin_{\lambda\geq0} \left\{ H(\lambda\mid B)-\lambda\psi \right\}.
\]
The next corollary uses this dual representation to describe how the shadow value changes with the required incentive level and with pledgeable execution rewards. The first statement is global and does not require differentiability. The derivative formulas are local statements that apply at smooth points of the dual objective, away from score ties and bonus-activation thresholds.


\begin{corollary}
\label{prop:lambda_comparative_statics}

The set $\Lambda(\psi\mid B)$ is nondecreasing in the execution cost $\psi$. If $\psi_2>\psi_1$, then every $\lambda_1\in\Lambda(\psi_1\mid B)$ and $\lambda_2\in\Lambda(\psi_2\mid B)$ satisfy $\lambda_2\geq\lambda_1$.

At points where the minimizer $\lambda^\ast(a)$ is unique, interior, locally differentiable, and the dual objective is twice differentiable with $H_{\lambda\lambda}(\lambda^\ast(a)\mid a)>0$, the following local comparative statics hold under the common scaling $B_i(a)=aB_i$,
\[
\frac{d\lambda^\ast}{d\psi} = \frac{1}{H_{\lambda\lambda}(\lambda^\ast(a)\mid a)} >0 \; \text{ and }\; \frac{d\lambda^\ast}{da} = -\frac{H_{\lambda a}(\lambda^\ast(a)\mid a)} {H_{\lambda\lambda}(\lambda^\ast(a)\mid a)}.
\]
In particular, $d\lambda^\ast/da\leq 0$ whenever $H_{\lambda a}(\lambda^\ast(a)\mid a)\geq 0$.
\end{corollary}

Because $H(\lambda\mid B)$ is defined using maximum operators and the positive-part term in $\chi_i(\lambda)$, it need not be smooth everywhere. Kinks arise, for example, when a bidder's success-contingent bonus becomes just worth using, or when two allocation choices give the same shadow surplus. At such points, the multiplier may stay fixed while the mechanism adjusts the allocation or the bonus expenditures to satisfy the executor's incentive constraint exactly. Thus, the global monotonicity result is the robust comparative static. The derivative formulas describe the smoother regions between kinks.

Corollary~\ref{prop:lambda_comparative_statics} shows that when the required execution incentive is higher, the shadow value of the executor's constraint is higher, and shadow scores distort the allocation more strongly. When pledgeable bonuses become more abundant, the same incentive can be provided with a lower shadow value, so the allocation distortion weakly falls under the stated regularity conditions.

The condition $H_{\lambda a}(\lambda^\ast(a)\mid a)\geq 0$ says that, at a given shadow value, larger pledgeable rewards weakly increase the amount of execution incentive generated by the shadow-score allocation. When this holds, more pledgeability reduces the shadow value needed to meet a fixed incentive requirement. If rewards become easier to pledge, the mechanism can motivate the executor with a smaller distortion to the allocation rule.

The preceding comparative statics describe how the optimal multiplier changes when the execution problem becomes more or less demanding. The dual representation also gives the multiplier a direct value  interpretation. Since $\lambda^\ast$ is the shadow value of the execution agent's incentive constraint, increasing the required incentive level lowers the seller's effort-inducing payoff at rate $\lambda^\ast$. Similarly, increasing a bidder's pledgeable success bonus raises the value of the optimal mechanism
only in states where that bidder wins and where its execution reward has positive shadow value. The next result shows these envelope values.

\begin{corollary}
\label{prop:envelope_values}
Let $V^E(\psi,B)$ denote the seller's value from the optimal mechanism conditional on inducing execution effort, where $B=(B_1,\ldots,B_n)$. Suppose that the strict-feasibility condition \eqref{eq:slater} holds, the optimal multiplier is unique, and the score maximizer is unique almost everywhere. Then, at points of differentiability,
\[
\frac{\partial V^E}{\partial \psi}=-\lambda^\ast.
\]
Moreover, for each bidder $i$,
\[
\frac{\partial V^E}{\partial B_i} = \E\left[ x_i^\ast(v)\max\{\lambda^\ast\D_i-\sigma_i,0\} \right].
\]
\end{corollary}

Corollary~\ref{prop:envelope_values} gives two economic interpretations of the shadow-score mechanism. The first derivative shows that the multiplier $\lambda^\ast$ is the marginal cost of making the execution incentive constraint harder to satisfy. If the required incentive level $\psi$ rises by one unit, the seller's optimal effort-inducing value falls at rate $\lambda^\ast$. The same multiplier that distorts bidder rankings also measures the value of relaxing the execution agent's moral hazard constraint.

The second derivative measures the marginal value of pledgeability. Increasing $B_i$ expands the amount of success-contingent reward that can be promised when bidder $i$ wins. This is valuable only in states in which bidder $i$ receives the object and only when bidder $i$'s execution reward has positive shadow value, that is, when $\lambda^\ast\Delta_i>\sigma_i$. The expression $\E\left[ x_i^\ast(v)\max\{\lambda^\ast\Delta_i-\sigma_i,0\} \right]$ has a basic interpretation. It is the expected shadow surplus from relaxing bidder $i$'s pledgeability constraint. Pledgeable rewards are most valuable for bidders who win often under the optimal allocation and whose allocation states are cost effective instruments for inducing execution effort.


\subsubsection{The no-effort alternative}

The characterization in Theorem~\ref{thm:execution_score} is conditional on inducing execution effort. The seller's global problem also includes the possibility of not inducing effort. If effort is not induced, success-contingent rewards provide no incentive benefit and are costly, so the seller sets them equal to zero. Let $g_i^0$ denote the seller's execution payoff when bidder $i$ receives the object and the execution agent does not exert effort. The resulting problem is the standard virtual-surplus problem with no-effort execution score
\[
S_i^0(v_i)\coloneqq \varphi_i(v_i)+g_i^0.
\]
Under Assumption~\ref{ass:regularity}, the optimal no-effort mechanism assigns the object to a bidder with the highest nonnegative value of $S_i^0(v_i)$ and is implemented by critical value payments.

Let $W^E$ be the seller's value from the optimal effort-inducing mechanism and let $W^0$ be the seller's value from the optimal no-effort mechanism. The seller induces execution effort if and only if $W^E\geq W^0$. If this inequality fails, the seller uses the no-effort virtual-surplus allocation with score $S_i^0(v_i)$. In the remainder of the paper, I focus on the case in which inducing execution effort is optimal.


\subsubsection{Post-win execution effort}
\label{subsec:post_win_effort}

The baseline model assumes that execution effort is chosen before bidder values are realized and before the winner is known. This timing captures ex ante execution capacity, such as monitoring resources, enforcement capacity, review teams, or compliance infrastructure must be prepared before the identity of the winner is known. It is useful to contrast this case with the alternative timing in which the execution agent chooses effort only after observing the winner.

Suppose, in this section only, that effort is chosen after the auction outcome is realized. If bidder $i$ wins, the execution agent exerts effort if and only if $\D_i b_i\geq \psi$. Thus, effort can be induced after assigning the object to bidder $i$ if and only if $\D_i B_i\geq \psi$. When this condition holds and $\D_i>0$, the cheapest effort-inducing success-contingent bonus is
\[
b_i^E\coloneqq\frac{\psi}{\D_i}.
\]
Let $g_i^0$ denote the seller's execution payoff when bidder $i$ receives the object and effort is not exerted. Define the post-win execution value of assigning the object to bidder $i$ by
\[
\Gamma_i \coloneqq \max\left\{ g_i^0,\, g_i-\sigma_i\frac{\psi}{\D_i} \right\},
\]
where the effort term is available only if $\D_i B_i\geq\psi$. If $\D_i=0$ or $\D_i B_i<\psi$, effort is infeasible and we set $\Gamma_i=g_i^0$.

The execution problem is no longer an aggregate constraint on the allocation rule. After each possible winner, the seller simply chooses whether inducing effort is worthwhile and feasible. Lemma~\ref{prop:post_win_effort} shows that this collapses the auction problem to a standard virtual-surplus maximization with winner-specific continuation values.

\begin{lemma}
\label{prop:post_win_effort}
Suppose the execution agent chooses effort after observing the winner. Under Assumption~\ref{ass:regularity}, the optimal mechanism assigns the object to a bidder with the highest nonnegative score
\[
S_i^{PW}(v_i)=\varphi_i(v_i)+\Gamma_i.
\]
The allocation is implementable by a dominant-strategy scoring auction with critical value payments.
\end{lemma}

This comparison clarifies the role of the baseline timing assumption. With post-win effort, execution affects the auction only through winner-specific continuation values. There is no aggregate moral hazard constraint linking the allocation rule to the executor's ex ante incentive to exert effort. The multiplier-based shadow-score channel in Theorem~\ref{thm:execution_score} is specific to ex ante execution capacity.

Ex ante effort is natural when execution requires capacity to be prepared before the winner is known, such as monitoring systems, administrative review, enforcement capacity, or platform compliance resources. Post-win effort is natural when execution can be tailored after the auction outcome, such as due diligence, winner-specific contracting, or case-specific approval.

\subsection{The shadow-score auction format}
\label{sec:execution}

Theorem~\ref{thm:execution_score} characterizes the optimal effort-inducing allocation. This section describes the corresponding auction format. The mechanism is second-price-like, but bidders are ranked by shadow scores rather than by values alone.

For the optimal multiplier $\lambda^\ast$, define bidder $i$'s \emph{shadow handicap} by
\begin{equation}\label{eq:handicap}
h_i^\ast\coloneqq g_i+\chi_i(\lambda^\ast).
\end{equation}
This handicap has two components. The first, $g_i$, is the direct execution payoff from assigning the object to bidder $i$. The second, $\chi_i(\lambda^\ast)$, is the endogenous shadow value of using bidder $i$'s allocation state to relax the executor's incentive constraint. Bidder $i$'s shadow score is thus $S_i(v_i\mid\lambda^\ast)=\varphi_i(v_i)+h_i^\ast$. The object is assigned to a bidder with the highest nonnegative shadow score, with ties broken by a fixed exogenous priority rule. Ignoring tie profiles, bidder $i$ wins whenever
\[
\underbrace{\varphi_i(v_i)+h_i^\ast}_{=S_i(v_i\mid \lambda^\ast)} \geq \max\left\{ 0,\max_{j\neq i}\left(\varphi_j(v_j)+h_j^\ast\right) \right\}.
\]
Let $M_i(v_{-i}\mid \lambda^\ast)$ be as defined by \eqref{eq:m}. Since $\varphi_i$ is strictly increasing, bidder $i$'s critical value is
\[
p_i(v_{-i}) = \inf\left\{ z\in[\ubar v_i,\vbar v_i] \mid \varphi_i(z)+h_i^\ast \geq M_i(v_{-i}\mid \lambda^\ast) \right\},
\]
where the critical value is kept within bidder $i$'s value support. A winning bidder pays this critical value.

The payment rule has the same logic as a second-price auction. That is, a bidder does not pay its own report, but the lowest report at which it would still win. The difference is the relevant comparison. In a standard second-price auction with a reserve, the allocation compares values and uses the reserve only to decide whether the object is sold. In a shadow-score auction, the allocation compares $\varphi_i(v_i)+h_i^\ast$. A bidder with a larger handicap can win with a lower value because that bidder either generates a higher direct execution payoff or provides a more valuable state for satisfying the executor's incentive constraint. 

Figure~\ref{fig:shadow_score_ranking_reversal} illustrates the ranking channel created by the shadow score. A bidder with a higher execution premium has a score curve that lies above that of an otherwise identical bidder. As a result, the shadow-score auction may assign the object to a bidder with a lower value when that bidder's allocation state is more useful for satisfying the execution agent's incentive constraint.\footnote{The figure uses two bidders with values independently distributed on $[0,1]$, so $\varphi_i(v_i)=2v_i-1$. The execution primitives are $g_1=g_2=0$, $B_1=1$, $B_2=0$, $\Delta_1=\sigma_1=1/2$, $\Delta_2=0$, and $\psi=23/80$. These parameters imply $\lambda^\ast=9/5$. Hence, $\chi_1(\lambda^\ast)=B_1(\lambda^\ast\Delta_1-\sigma_1)=2/5$ and $\chi_2(\lambda^\ast)=0$, giving $S_1(v)=2v-1+2/5$ and $S_2(v)=2v-1$.}


\begin{figure}
    \centering
    \includegraphics[width=0.75\linewidth]{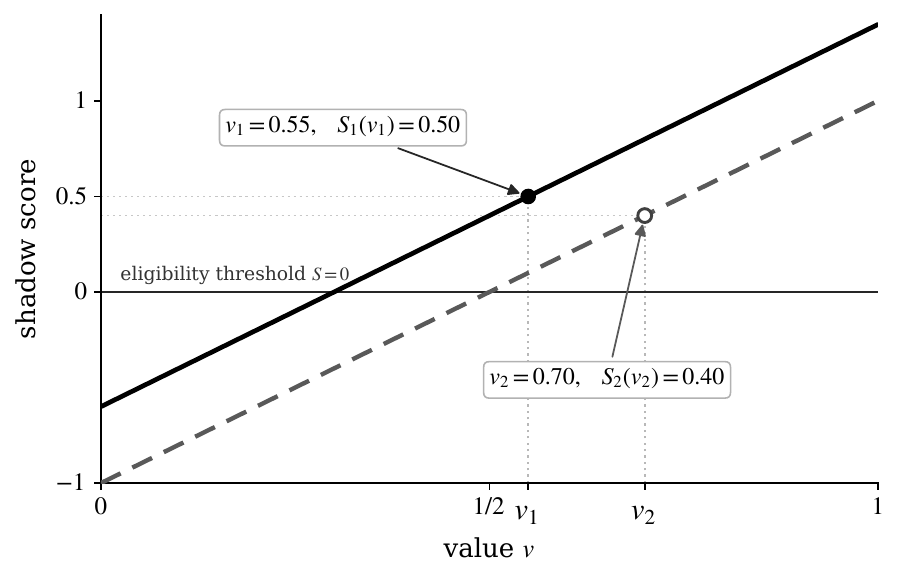}
    \caption{Shadow scores and ranking reversal. The solid line is bidder 1's shadow score, while the dashed line is bidder 2's shadow score. The horizontal line marks the eligibility threshold $S=0$. At the displayed values, bidder 1 has a lower value, $v_1<v_2$, but a higher shadow score, $S_1(v_1)>S_2(v_2)$. As a result, the shadow-score auction assigns the object to bidder 1.}
\label{fig:shadow_score_ranking_reversal}
\end{figure}


If bidders are symmetric on the value side and all shadow handicaps are equal, the shadow-score ranking coincides with the value ranking. In this case, the mechanism reduces to an anonymous second-price auction with a common reserve. However, if shadow handicaps differ, the reserve-auction logic no longer applies. The optimal mechanism may change the ranking of eligible bidders, not only the threshold for sale.

The seller's objective is revenue plus execution payoff, net of bonus costs. For policy applications, it is also useful to ask how the shadow-score allocation looks from a total-surplus perspective. Ignoring transfers, let
\[
W^E(x) \coloneqq \E\left[ \sum_{i=1}^n x_i(v)\left(v_i+g_i\right) \right] -\psi
\]
denote total surplus from allocation rule $x$ when effort is induced. Moreover, let
\[
W^0(x) \coloneqq \E\left[ \sum_{i=1}^n x_i(v)\left(v_i+g_i^0\right) \right]
\]
denote total surplus without effort. Let $x^{S}$ be the shadow-score allocation, let $x^{FB,E}$ maximize $W^E(x)$, and let $x^{FB,0}$ maximize $W^0(x)$. Then,
\[
W^E\left(x^{S}\right)-W^0\left(x^{FB,0}\right) = \underbrace{ \left[W^E(x^{FB,E})-W^0(x^{FB,0})\right] }_{\text{potential surplus gain from effort}} - \underbrace{ \left[W^E(x^{FB,E})-W^E(x^{S})\right] }_{\text{misallocation loss under shadow scoring}}.
\]
The first term is the surplus gain that effort could create under the first-best allocation. The second term is the surplus loss from using the revenue- and incentive-optimal shadow-score allocation instead. Shadow scoring is valuable when the execution gains it makes possible are large relative to this misallocation loss.



\section{Format restrictions and the value of shadow scoring}
\label{sec:format_restrictions}

We have seen that the optimal effort-inducing mechanism is a shadow-score auction. In practice, however, sellers may be restricted to simpler formats. Procurement agencies, license authorities, and platforms often use anonymous second-price auctions with a common reserve, or at most bidder-specific reserve prices based on observable bidder categories. This section compares these restricted formats with the shadow-score allocation.

The shadow-score auctions' ranking channel is essential when execution incentives depend on the identity of the winner. To isolate the comparison, suppose in this section that bidders are symmetric on the value side, i.e., $F_i=F$ for all $i\in N$, with common support $[\ubar v,\vbar v]$ and common continuous, strictly increasing virtual value $\varphi$. The shadow-score auction assigns the object to a bidder with the highest nonnegative shadow score, $S_i(v_i)$, with handicap $h_i^\ast$ as defined by \eqref{eq:handicap}.

We compare this allocation with two restricted formats. An anonymous second-price auction with a common reserve uses a single threshold $r$ and assigns the object to the highest-value bidder among those with $v_i\geq r$. A second-price auction with bidder-specific reserves uses thresholds $(r_1,\ldots,r_n)$ and assigns the object to the highest-value bidder among those with $v_i\geq r_i$.

We use the following terminology. A bidder $i$ is \emph{active} under the shadow-score allocation if it wins with positive probability, i.e., if $\Pr(x_i^\ast(v)>0)>0$. A restricted reserve mechanism is \emph{effort-feasible} if it can be paired with bonus expenditures that induce execution effort.

For the payoff comparison, let $y=(x,q)$ denote an allocation-bonus pair, where $x=(x_1,\ldots,x_n)$ is an allocation rule and $q=(q_1,\ldots,q_n)$ satisfies $0\leq q_i(v)\leq B_i x_i(v)$ for every $i$ and $v$. Define the execution incentive generated by $y$ as
\[
I(y)\coloneqq \E\left[\sum_{i=1}^n \D_i q_i(v)\right],
\]
and define the seller's expected virtual payoff, net of execution bonuses, as
\[
\Pi(y)\coloneqq \E\left[ \sum_{i=1}^n x_i(v)\left(\varphi(v_i)+g_i\right) - \sum_{i=1}^n \sigma_i q_i(v) \right].
\]
An allocation-bonus pair is effort-feasible if $I(y)\geq\psi$. Restricted formats are allowed to choose their bonus expenditures subject to the same bounds $0\leq q_i(v)\leq B_i x_i(v)$. The comparison does not disadvantage reserve auctions by forcing them to use the shadow-score auction's bonus rule.

The shadow-score theorem characterizes the unconstrained optimal effort-inducing mechanism. A separate question is how much of this mechanism can be reproduced by familiar reserve formats. This question matters because many sellers are institutionally constrained to use anonymous reserves or bidder-specific eligibility rules. The comparison below shows that these formats miss the ranking channel created by execution incentives.

\begin{proposition}
\label{prop:format_restrictions}
Suppose values have common support $[\ubar v,\vbar v]$ and $\varphi$ is continuous and strictly increasing. Consider the optimal shadow-score allocation.

\begin{enumerate}[label=(\roman*)]

\item If all bidders have the same shadow handicap, i.e. $h_i^\ast=h_j^\ast$ for all $i,j\in N$, then the shadow-score allocation is implementable by an anonymous second-price auction with a common reserve.

Conversely, if two active bidders $i$ and $j$ have different shadow handicaps, $h_i^\ast\neq h_j^\ast$, then no anonymous second-price auction with a common reserve implements the shadow-score allocation.

\item A bidder-specific-reserve auction can implement the shadow-score allocation only if the shadow-score allocation never requires a pairwise ranking reversal among active eligible bidders. In particular, if two active bidders $i$ and $j$ have different shadow handicaps, $h_i^\ast\neq h_j^\ast$, then no bidder-specific-reserve auction implements the shadow-score allocation.

\item Let $y^\ast=(x^\ast,q^\ast)$ be an optimal shadow-score mechanism, and suppose the associated multiplier satisfies $\lambda^\ast>0$. Suppose also that shadow-score maximization is unique almost everywhere. If no effort-feasible restricted reserve mechanism has allocation rule $x^\ast$ almost everywhere, then every effort-feasible restricted reserve mechanism yields strictly lower expected payoff than the shadow-score auction, i.e., $\Pi(y)<\Pi(y^\ast)$.
\end{enumerate}
\end{proposition}

Proposition~\ref{prop:format_restrictions} shows that the limitation of reserve formats is structural. A common reserve can change only the sale threshold. Bidder-specific reserves can also change which bidders are eligible to win. Neither format can generally change the ranking of bidders who remain eligible. Shadow scores can do exactly that. When execution incentives make one bidder's allocation state more valuable than another's, the optimal mechanism may award the object to a lower-value eligible bidder because doing so relaxes the execution agent's incentive constraint. This ranking channel cannot generally be reproduced by reserves.

The restriction is payoff-relevant when the execution constraint has positive shadow value. In that case, any effort-feasible reserve format that cannot reproduce the shadow-score allocation leaves value on the table. The next example quantifies this loss in a two-bidder uniform environment.


\subsection{A two-bidder uniform example}
\label{subsec:uniform_example}

We now quantify the loss from restricting attention to common reserves in a simple example. There are two bidders with independent values uniformly distributed on $[0,1]$. Bidder 1 is execution-relevant and bidder 2 is not. The execution requirement is represented in reduced form by the constraint
\begin{equation}
\Pr(\text{bidder 1 wins})\geq \kappa .
\label{eq:kappa_constraint}
\end{equation}
There are no direct execution costs in this example. The virtual value is $\varphi(v)=2v-1$.

Without the execution constraint, the optimal auction is the standard Myerson auction \citep{Myerson1981}, i.e., a second-price auction with reserve $1/2$. Bidder 1's winning probability under that auction is
\[
\Pr(v_1\geq 1/2,\ v_1\geq v_2) = \int_{1/2}^1 v_1\,dv_1 = \frac{3}{8}.
\]
The standard Myerson auction violates
\eqref{eq:kappa_constraint} whenever $\kappa>3/8$.

\begin{proposition}
\label{prop:common_reserve_loss}
In the two-bidder uniform example, the following statements hold.

\begin{enumerate}[label=(\roman*)]

\item The standard Myerson auction satisfies the execution constraint if and only if $\kappa\leq 3/8$.

\item For $\kappa\in\left(3/8, 1/2\right]$, the best anonymous second-price auction with a common reserve uses a reserve $r^C(\kappa)=\sqrt{1-2\kappa}$, and yields expected virtual surplus
\[
R^C(\kappa) = \frac{1}{3} + \left(r^C(\kappa)\right)^2 - \frac{4}{3}\left(r^C(\kappa)\right)^3.
\]
For $\kappa>1/2$, no common-reserve auction can satisfy the execution constraint.

\item For $\kappa\in\left[3/8, 7/8\right]$, the optimal shadow-score auction gives bidder~1 a virtual-score premium equal to $2\left(\kappa-3/8\right)$. The virtual scores are
\[
S_1(v_1) = 2v_1-1+2\left(\kappa-\frac{3}{8}\right),
\]
\[
S_2(v_2) = 2v_2-1.
\]
The resulting expected virtual surplus is
\[
R^S(\kappa) = \frac{5}{12} - \left(\kappa-\frac{3}{8}\right)^2.
\]

\item For every $\kappa\in\left(3/8, 1/2\right]$, the shadow-score auction strictly dominates the best feasible common-reserve auction, as $R^S(\kappa)>R^C(\kappa)$. Over the range in which common reserves remain feasible, the loss from using the best common reserve rather than the shadow-score auction is maximized at $\kappa=1/2$, and
\[
\max_{\kappa\in[3/8,1/2]}\left( R^S(\kappa)-R^C(\kappa)  \right)=\frac{13}{192}.
\]
\end{enumerate}
\end{proposition}


\begin{figure}
    \centering
    \includegraphics[width=0.75\linewidth]{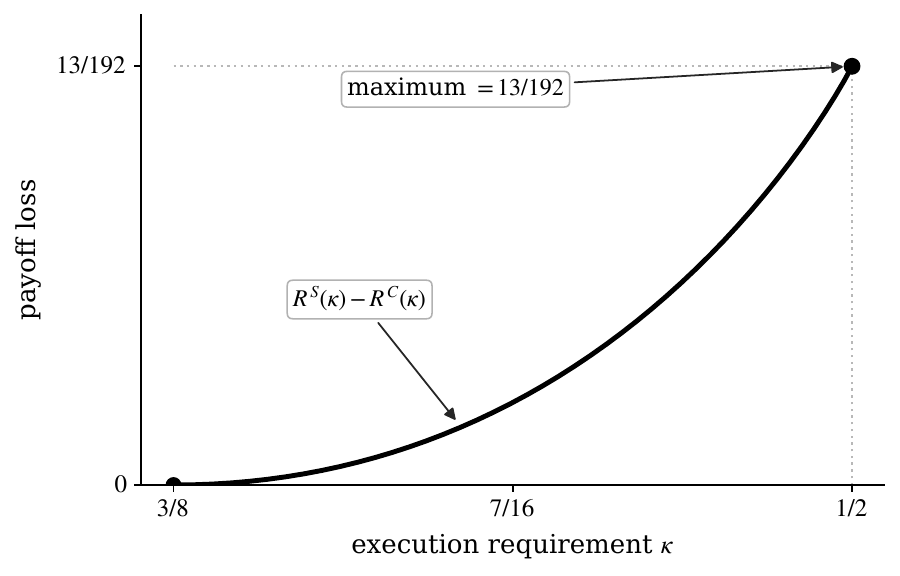}
    \caption{Loss from restricting attention to common reserves in the two-bidder uniform example. The figure plots the payoff loss over the range $\kappa\in[3/8,1/2]$, where common reserves are feasible and the execution constraint binds. At $\kappa=3/8$, the standard Myerson auction already satisfies the execution requirement, so the loss is zero. The loss is maximized at $\kappa=1/2$, where it equals $13/192$. For $\kappa>1/2$, no common-reserve auction satisfies the execution requirement.}
\label{fig:common_reserve_loss}
\end{figure}


The example separates two ways of satisfying the execution requirement. A common reserve can increase the probability that bidder 1 wins only by lowering the sale threshold. This becomes infeasible once $\kappa>1/2$. A shadow-score auction instead changes the pairwise ranking of the two bidders. It lets bidder 1 win even when bidder 2 has a slightly higher value, precisely because assigning the object to bidder 1 relaxes the execution requirement. This ranking channel is what reserve prices alone cannot reproduce.

Figure~\ref{fig:common_reserve_loss} illustrates how the cost of using common reserves varies with the execution requirement $\kappa$. In this example, $\kappa$ is the minimum probability with which bidder 1, the execution-relevant bidder, must receive the object. At $\kappa=3/8$, the unconstrained Myerson auction already satisfies the requirement, so the restriction to common reserves is costless. As $\kappa$ increases, a common reserve can raise bidder 1's winning probability only by lowering the sale threshold, whereas the shadow-score auction uses the ranking channel and gives bidder 1 a score premium. The loss from using the best common reserve rises over the feasible range and reaches $13/192$ at $\kappa=1/2$. For larger requirements, no common-reserve auction can satisfy the constraint.


\section{When reduced-form scoring breaks down}
\label{sec:beyond_scoring}

The baseline admits a reduced-form scoring interpretation. After the optimal multiplier is determined, the allocation can be written as if the seller had bidder-specific payoff shifters. This equivalence relies on two restrictions. Execution primitives must be public or certifiable, and execution effects must be winner-specific rather than report-dependent or set-dependent. This section studies departures from those restrictions.

The results separate four ways in which the standard scoring representation can fail. Private reliability restricts what can be elicited from bidders. Certification determines which execution characteristics become public enough to be scored. Execution externalities make the shadow value of assigning the object to one bidder depend on other bidders' reports. Set-dependent execution makes the relevant shadow value attach to a winner set rather than to individual bidders. In all four cases, the optimal auction is a constrained shadow allocation, not generally a standard scoring auction with primitive bidder qualities.

The characterizations in this section are less closed-form than Theorem~\ref{thm:execution_score}. In the baseline, the execution problem generates bidder-specific shadow premia, so the optimal allocation can be implemented by a shadow-score auction. In the environments below, the relevant execution information is not directly scoreable, or the execution value is not bidder-separable. The optimal mechanism is characterized as the solution to a constrained shadow-allocation problem. These results identify what replaces the baseline score auction and why the standard scoring representation is no longer available.


\subsection{A constrained shadow-allocation template}
\label{subsec:constrained_template}

The extensions in this section follow the same logic. Each one first asks what the mechanism is allowed to observe and use. In some cases the mechanism can use only value reports. In others it can use willingness-to-pay reports, public certification signals, or the full report profile. Once this information is fixed, the public execution primitives of the baseline are replaced by the corresponding effective primitives: posterior, signal-contingent, or report-dependent. The seller then solves the same kind of shadow-allocation problem as before, but now subject to the incentive-compatibility constraints appropriate for that environment.

Throughout the individual-allocation extensions, let $\mathcal X^{IC}$ denote the set of feasible possibly randomized allocation rules $x$ satisfying $\sum_{i=1}^n x_i(v)\leq 1$ and $0\leq x_i(v)\leq 1$ for every $v$, and such that, for each bidder $i$ and every $v_{-i}$, the allocation probability $x_i(v_i,v_{-i})$ is nondecreasing in $v_i$. Any allocation in $\mathcal X^{IC}$ is implementable in dominant strategies by the envelope payment formula
\begin{equation}\label{eq:envelope}
p_i(v_i,v_{-i}) = v_i x_i(v_i,v_{-i}) - \int_{\underline v_i}^{v_i}x_i(s,v_{-i})\,ds,
\end{equation}
with the lowest type receiving zero utility. For deterministic allocations, this formula reduces to critical value payments.

Two variants of this allocation set will also be used. In the payoff-relevant reliability case, the one-dimensional payoff type is willingness to pay rather than value. I denote by $\mathcal X^{IC,\theta}$ the set of feasible allocation rules that are monotone in each bidder's willingness to pay, denoted by $\theta_i$. In the public certification case, the allocation may also condition on the public signal profile $s$. I denote by $\mathcal X^{IC,\mathcal C}$ the set of feasible allocation rules $x(v,s)$ that are monotone in $v_i$, for each bidder $i$ and each fixed signal profile $s$.

The strict-feasibility conditions below are all versions of the same requirement. There must exist a feasible monotone allocation-bonus pair that generates more than the required execution incentive $\psi$. This condition ensures that the executor's incentive constraint admits a supporting multiplier. The resulting allocation problem is then a constrained shadow-allocation problem. If the pointwise maximizer of the shadow objective is monotone, the mechanism is a shadow-score auction. If it is not monotone, the optimal mechanism is the monotonicity-constrained, or ironed, version.

Assumption~\ref{ass:regularity} is used in the baseline to show that shadow-score maximization is monotone and truthfully implementable. In the applications analyzed by Section~\ref{sec:beyond_scoring}, pointwise maximization may fail monotonicity. The optimal mechanisms are characterized by maximizing shadow objectives over incentive compatible allocation sets. For these constrained characterizations, monotonicity of virtual values is not required for implementability, because monotonicity is imposed directly. Regularity is imposed only when a result claims that an unconstrained score rule, or set-level score rule, is itself implementable.


\subsection{Private payoff-irrelevant reliability}
\label{sec:private}

The analysis so far assumes that the execution parameters $\D_i$, $\sigma_i$, $B_i$, and $g_i$ are public. This is natural when execution capacity is certified, audited, or otherwise observable before the auction. In many applications, however, the relevant reliability of a bidder is partly private. A firm may know more than the seller about its ability to comply with deployment obligations, its internal organization, the quality of its subcontractors, its exposure to litigation, or the ease with which a regulator can monitor it.

Private reliability creates a distinct problem when it matters for execution but not for the bidder's own payoff. A bidder may prefer winning at a low price independently of whether it is easy or difficult for the regulator to monitor. In that case, reliability is not a standard private value. It is a private characteristic that affects the seller's execution problem but not the bidder's preferences. This section shows that such information cannot simply be inserted into a shadow score and elicited as if it were an ordinary payoff-relevant type. Incentive compatibility forces reliability types with the same value to receive the same interim treatment. I then study the robust design problem in which the mechanism does not condition allocations or execution rewards on unverifiable reliability reports.

The payoff-irrelevance assumption applies when a bidder cares about winning and paying, but does not directly internalize how easy it is for the execution agent to monitor, regulate, or coordinate with it. In applications where reliability also affects the bidder's continuation payoff, reputation, penalties, or operating profits, reliability becomes payoff-relevant and the bunching result below need not apply in this form (see Section~\ref{sec:relevant}).

Suppose bidder $i$'s type is $(v_i,a_i)$, where $v_i$ is its private value and $a_i\in A_i$ is a reliability parameter. The reliability parameter affects execution primitives, such as $\D_i(a_i)$, $\sigma_i(a_i)$, $B_i(a_i)$, and $g_i(a_i)$, but it does not enter bidder $i$'s payoff directly. If bidder $i$ receives the object and pays $p_i$, its utility is $v_i-p_i$, independently of $a_i$. A direct mechanism asks bidder $i$ to report both $v_i$ and $a_i$.

For a direct mechanism, let $X_i(v_i,a_i)$ and $P_i(v_i,a_i)$ denote bidder $i$'s interim allocation probability and expected payment when it reports $(v_i,a_i)$ truthfully, taking expectations over the other bidders' types. Let $\mathcal T_i\subseteq [\underline v_i,\overline v_i]\times A_i$ denote bidder $i$'s type support. For a reliability type $a_i$, write
\[
V_i(a_i)\coloneqq\{v_i \mid (v_i,a_i)\in\mathcal T_i\}
\]
for the set of values at which reliability type $a_i$ can occur. The bunching result below compares reliability types only at values where both types are in the support. This qualification matters because reliability may be correlated with value. If some reliability types occur only at some values, then the value report itself may contain information about reliability.

\begin{proposition}
\label{prop:private_bunching}
Suppose reliability $a_i$ is payoff-irrelevant for bidder $i$. In any Bayesian incentive compatible mechanism, for every two reliability types $a_i,a_i'\in A_i$, we have $X_i(v_i,a_i)=X_i(v_i,a_i')$ and $P_i(v_i,a_i)=P_i(v_i,a_i')$ for almost every $v_i\in V_i(a_i)\cap V_i(a_i')$. Conditional on any value at which both reliability types can occur, the two types are bunched. The mechanism cannot use unverifiable, payoff-irrelevant reliability reports to give different interim allocation probabilities or expected payments to reliability types that share the same value.
\end{proposition}

Proposition~\ref{prop:private_bunching} is a conditional bunching result. It says that, at a value where two reliability types can both occur, incentive compatibility forces them to receive the same interim allocation probability and expected payment. However, it does not claim that reliability is irrelevant when it is statistically related with value. In that case, the value report may reveal information about reliability, and the designer can use the posterior execution values associated with the reported one.

The bunching result rules out a simple extension of the baseline shadow-score auction. If reliability were public, the designer could insert $a_i$-dependent execution parameters into the score and allocate accordingly. When reliability is private, unverifiable, and payoff-irrelevant, this is not possible. Proposition~\ref{prop:private_bunching} implies that all reliability types with the same value must receive the same interim allocation and payment. Hence, any public reliability shadow score that separates such types cannot be implemented by asking bidders to report reliability. The next corollary makes this implication explicit.

\begin{corollary}
\label{cor:public_score_unattainable}
 Consider the shadow-score allocation that would be optimal if reliability were public. If, for some bidder $i$, two reliability types $a_i,a_i'$, and a positive measure set of values $V\subseteq V_i(a_i)\cap V_i(a_i')$, the public reliability allocation induces different interim allocation probabilities, i.e., $X_i^P(v_i,a_i)\neq X_i^P(v_i,a_i')$ for $v_i\in V$, then no Bayesian incentive compatible direct mechanism can implement that allocation using only unverifiable reports of $a_i$.
\end{corollary}

The designer may still use reliability information if it is publicly certified, or if it is statistically related to the value report. The second case creates a posterior design problem. I focus on mechanisms that do not condition on unverifiable reliability reports themselves. Under this restriction, the auction can use only the information about reliability contained in the bidder's value report.

For the rest of this section, suppose that the bonus cap is public and does not depend on unverifiable reliability. Denote it by $\bar B_i$. For a fixed value report $v_i$, define the posterior execution parameters $\bar\Delta_i(v_i) \coloneqq  \E[\Delta_i(a_i)\mid v_i]$, $\bar\sigma_i(v_i) \coloneqq  \E[\sigma_i(a_i)\mid v_i]$, and $\bar g_i(v_i) \coloneqq  \E[g_i(a_i)\mid v_i]$. The posterior incentive generated by an allocation-bonus pair $(x,q)$ is \[ \bar I(x,q) \coloneqq  \E\left[ \sum_{i=1}^n \bar\Delta_i(v_i)q_i(v) \right], \] where $0\leq q_i(v)\leq \bar B_i x_i(v)$. I impose the corresponding strict feasibility condition,
\begin{equation}
\exists\,(\tilde x,\tilde q) \quad\text{such that}\quad \tilde x\in\mathcal X^{IC},\quad 0\leq \tilde q_i(v)\leq \bar B_i\tilde x_i(v), \quad\text{and}\quad \bar I(\tilde x,\tilde q)>\psi.
\label{eq:private_slater}
\end{equation}

The strict-feasibility condition ensures that the posterior execution problem admits a supporting multiplier. The next proposition characterizes the optimal effort-inducing mechanism in this environment. The logic parallels Theorem~\ref{thm:execution_score}, but with an important difference. Since unverifiable reliability cannot be conditioned on directly, the relevant execution parameters are posterior objects, conditional on the reported value. The resulting posterior shadow score need not be increasing in value. Hence, the designer cannot generally allocate pointwise to the highest posterior shadow score. The optimal allocation instead maximizes posterior shadow surplus subject to the monotonicity constraints required for incentive compatibility.

For each bidder $i$, define the posterior execution value of assigning the object to bidder $i$ by
\[
\bar\chi_i\left(v_i\mid\lambda^P\right) \coloneqq  \bar B_i \max\left\{\lambda^P\bar\Delta_i(v_i)-\bar\sigma_i(v_i),0\right\}.
\]
The corresponding posterior shadow score is
\[
\bar S_i\left(v_i\mid\lambda^P\right) \coloneqq  \varphi_i(v_i)+\bar g_i(v_i)+\bar\chi_i\left(v_i\mid\lambda^P\right).
\]
This score is the analogue of the shadow score in Theorem~\ref{thm:execution_score}, with public execution parameters replaced by their posterior conditional on the reported value. Unlike the benchmark score, however, $\bar S_i(v_i\mid\lambda)$ need not be increasing in $v_i$. Hence, maximization of posterior shadow scores may violate incentive compatibility.

\begin{proposition}
\label{prop:optimal_private_reliability}
Suppose reliability is private, unverifiable, and payoff-irrelevant, and suppose the mechanism cannot condition on unverifiable reliability reports. Suppose \eqref{eq:private_slater} holds. Then, there exists an optimal effort-inducing mechanism and a multiplier $\lambda^P\geq0$ such that the optimal allocation solves the monotonicity-constrained problem
\[
x^P \in \argmax_{x\in\mathcal X^{IC}} \E\left[ \sum_{i=1}^n x_i(v)\bar S_i(v_i\mid\lambda^P) \right].
\]
The associated bonus expenditure satisfies
\[
q_i^P(v) \in \argmax_{0\leq q\leq \bar B_i x_i^P(v)} \left(\lambda^P\bar\Delta_i(v_i)-\bar\sigma_i(v_i)\right)q.
\]
The allocation is implementable in dominant strategies by the envelope payment formula. If the maximizer of posterior shadow scores is monotone, then the optimal mechanism is the posterior shadow-score auction. If the maximizer is not monotone, the optimal mechanism is its monotonicity-constrained, or ironed, version.
\end{proposition}

Proposition~\ref{prop:optimal_private_reliability} shows how the scoring problem changes under this robustness restriction. The designer does not score self-reported reliability directly. Instead, the mechanism uses posterior execution values inferred from the value report.

The next example shows that the monotonicity constraint in Proposition~\ref{prop:optimal_private_reliability} can bind even with one bidder. The example uses a single bidder only to isolate the monotonicity issue. With one bidder, the implementable analogue of a second-price auction with a reserve is a posted-price mechanism, where the bidder wins if its value exceeds a threshold and pays that threshold. The example shows that posterior shadow-score maximization can fail incentive compatibility even before strategic interaction among bidders is introduced.\footnote{See Appendix~\ref{app:beyond} for details about the example.}

\begin{example}
\label{prop:posterior_nonmonotone}
There exists a one-bidder environment with regular values and private, payoff-irrelevant reliability in which the shadow-score maximizer is not incentive compatible. In that environment, the optimal incentive compatible mechanism cannot implement the positive posterior-score rule. It must instead iron or discard part of the execution value.
\end{example}


\subsection{Private payoff-relevant reliability}\label{sec:relevant}

Section~\ref{sec:private} considers reliability that affects execution but not the bidder's own payoff. I now allow reliability to affect the bidder's payoff through willingness to pay. This payoff relevance changes the one-dimensional type from $v_i$ to $\theta_i$, but it does not make residual reliability separately scoreable when preferences over winning and paying are summarized by $\theta_i$.

Let bidder $i$'s private type be $(\theta_i,a_i)$, where $\theta_i$ is its willingness to pay and $a_i$ is an execution-reliability characteristic. If bidder $i$ receives the object and pays $p_i$, its payoff is $\theta_i-p_i$. This formulation includes, for example, a value $v_i$ and a reliability-dependent compliance cost $c_i(a_i)$, with $\theta_i=v_i-c_i(a_i)$. The reliability characteristic $a_i$ may affect the execution parameters, but the bidder's preferences over winning and paying are indexed by the scalar $\theta_i$. For a direct mechanism, let $X_i(\theta_i,a_i)$ and $P_i(\theta_i,a_i)$ denote bidder $i$'s interim allocation probability and expected payment.

The preceding bunching logic also applies when reliability affects the bidder's payoff only through willingness to pay. In that case the scalar $\theta_i$ is the payoff-relevant type. Conditional on $\theta_i$, any remaining reliability information is payoff-irrelevant, even if it matters for execution. The next corollary shows this implication.

\begin{corollary}
\label{cor:wtp_bunching}
Suppose bidder $i$'s payoff from receiving the object and paying $p_i$ is $\theta_i-p_i$, where $\theta_i$ is its willingness to pay. Let $a_i$ be a private, unverifiable reliability characteristic that may affect execution primitives. In any Bayesian incentive compatible mechanism, for every two reliability types $a_i,a_i'$ such that $(\theta_i,a_i)$ and $(\theta_i,a_i')$ belong to the type support, we have that $X_i(\theta_i,a_i)=X_i(\theta_i,a_i')$ and $P_i(\theta_i,a_i)=P_i(\theta_i,a_i')$ for almost every $\theta_i$. Thus, conditional on willingness to pay, residual reliability is bunched.
\end{corollary}

The mechanism can elicit willingness to pay, but not residual reliability among types with the same $\theta_i$. The optimal mechanism must use posterior execution parameters conditional on $\theta_i$. I assume that the willingness-to-pay types $\theta_i$ are independently distributed across bidders, with support $[\underline\theta_i,\overline\theta_i]$, distribution $F_i^\theta$, and density $f_i^\theta>0$. Define the virtual willingness to pay by
\[
\varphi_i^\theta(\theta_i) \coloneqq  \theta_i - \frac{1-F_i^\theta(\theta_i)}{f_i^\theta(\theta_i)}.
\]
Assume that $\varphi_i^\theta$ is strictly increasing.

Suppose the bonus cap is public and equal to $\bar B_i$. Define posterior execution parameters conditional on willingness to pay as $\bar\Delta_i(\theta_i) \coloneqq  \E[\Delta_i(a_i)\mid \theta_i]$, $\bar\sigma_i(\theta_i) \coloneqq  \E[\sigma_i(a_i)\mid \theta_i]$, and $\bar g_i(\theta_i) \coloneqq  \E[g_i(a_i)\mid \theta_i]$. For an allocation-bonus pair $(x,q)$, let \[ \bar I^\theta(x,q) \coloneqq  \E\left[ \sum_{i=1}^n \bar\Delta_i(\theta_i)q_i(\theta) \right], \] with $0\leq q_i(\theta)\leq \bar B_i x_i(\theta)$. As before, impose strict feasibility,
\begin{equation}
\exists\,(\tilde x,\tilde q) \quad\text{such that}\quad \tilde x\in\mathcal X^{IC,\theta},\quad 0\leq \tilde q_i(\theta)\leq \bar B_i\tilde x_i(\theta), \quad\text{and}\quad \bar I^\theta(\tilde x,\tilde q)>\psi,
\label{eq:payoff_relevant_private_slater}
\end{equation}
where $\mathcal X^{IC,\theta}$ is the set of feasible allocation rules monotone in each bidder's willingness to pay.\footnote{ The sets $\mathcal X^{IC}$ and $\mathcal X^{IC,\theta}$ impose the same incentive-compatibility requirement on different payoff types. The set $\mathcal X^{IC}$ contains feasible allocation rules that are monotone in reported values $v_i$. The set $\mathcal X^{IC,\theta}$ contains feasible allocation rules that are monotone in reported willingness to pay $\theta_i$. The notation changes only because, in the payoff-relevant reliability case, $\theta_i$ rather than $v_i$ is the one-dimensional payoff type. } For each bidder $i$, define 
\[ 
\bar\chi_i^\theta\left(\theta_i\mid\lambda^\theta\right) \coloneqq  \bar B_i \max\left\{\lambda^\theta\bar\Delta_i(\theta_i) -\bar\sigma_i(\theta_i),0\right\}, 
\]
and 
\[ 
\bar S_i^\theta\left(\theta_i\mid\lambda^\theta\right) \coloneqq  \varphi_i^\theta(\theta_i) + \bar g_i(\theta_i) + \bar\chi_i^\theta\left(\theta_i\mid\lambda^\theta\right). 
\]

The score $\bar S_i^\theta$ is the payoff-relevant analogue of the posterior shadow score in the preceding Section~\ref{sec:private}. Corollary~\ref{cor:wtp_bunching} implies that the mechanism can elicit willingness to pay, but cannot separately elicit residual reliability among types with the same $\theta_i$. Thus, the designer must treat $\theta_i$ as the one-dimensional payoff type and evaluate execution parameter through their posterior values conditional on $\theta_i$. The resulting problem has the same structure as the posterior reliability problem above, except that monotonicity is imposed with respect to willingness to pay rather than value. The next corollary shows this constrained posterior-shadow allocation.\footnote{The relevant single-dimensional type is $\theta_i$, not the underlying pair $(v_i,a_i)$. The envelope formula and the monotonicity constraint are imposed in willingness to pay.}


\begin{corollary}
\label{prop:optimal_payoff_relevant_reliability}
Suppose reliability is private and unverifiable, and suppose bidder preferences are indexed by willingness to pay $\theta_i$. Suppose \eqref{eq:payoff_relevant_private_slater} holds. Then, there exists an optimal effort-inducing mechanism and a multiplier $\lambda^\theta\geq0$ such that the optimal allocation solves
\[
x^\theta \in \argmax_{x\in\mathcal X^{IC,\theta}} \E\left[ \sum_{i=1}^n x_i(\theta) \bar S_i^\theta\left(\theta_i\mid\lambda^\theta\right) \right].
\]
The associated bonus expenditure satisfies
\[
q_i^\theta(\theta) \in \argmax_{0\leq q\leq \bar B_i x_i^\theta(\theta)} \left(\lambda^\theta\bar\Delta_i(\theta_i) -\bar\sigma_i(\theta_i)\right)q.
\]
The allocation is implemented by the envelope payment formula in willingness to pay. If the maximizer of posterior shadow scores is monotone in $\theta_i$, the optimal mechanism is the posterior shadow-score auction. If not, the optimal mechanism is its monotonicity-constrained, or ironed, version.
\end{corollary}


This result distinguishes payoff-relevant reliability from standard scoring. Even when reliability affects a bidder's willingness to pay, the auction can use only the information about execution contained in willingness to pay unless reliability is verified. The optimal mechanism is a posterior shadow allocation, and monotonicity may force ironing. A  standard scoring model would take the relevant quality as observable or reportable. In this environment, the quality component is not directly scoreable.


\subsection{Certification and scoreability}

The distinction between public and private reliability gives certification a key role. Certification, auditing, pre-qualification, and disclosure rules determine which execution characteristics can enter the allocation rule as public information. This section formalizes this idea.

Let a certification technology $\mathcal C$ generate a public signal $s_i\in S_i$ about bidder $i$'s reliability before the auction. The signal profile $s=(s_1,\ldots,s_n)$ is observed by the seller and by all bidders before bids are submitted. The technology has an ex ante cost $K(\mathcal C)$. Conditional on $s$, values are independent across bidders, and bidder $i$'s conditional distribution depends only on $s_i$. Let $F_i(\cdot\mid s_i)$ denote this conditional distribution, with density $f_i(\cdot\mid s_i)>0$. The corresponding conditional virtual value is
\[
\varphi_i(v_i\mid s_i) \coloneqq  v_i - \frac{1-F_i(v_i\mid s_i)}{f_i(v_i\mid s_i)}.
\]
Assume that $\varphi_i(\cdot\mid s_i)$ is strictly increasing for every signal $s_i$. For simplicity, suppose the bonus cap $B_i$ is public.

The certified signal also affects execution primitives. Conditional on $(v_i,s_i)$, define the scores $\Delta_i^{\mathcal C}(v_i,s_i)\coloneqq \mathbb E[\Delta_i(a_i)\mid v_i,s_i]$, $\sigma_i^{\mathcal C}(v_i,s_i)\coloneqq \mathbb E[\sigma_i(a_i)\mid v_i,s_i]$, and $g_i^{\mathcal C}(v_i,s_i)\coloneqq \mathbb E[g_i(a_i)\mid v_i,s_i]$.

Let $\mathcal X^{IC,\mathcal C}$ denote the set of feasible allocation rules under certification technology $\mathcal C$. An allocation rule $x=(x_1,\ldots,x_n)$ belongs to $\mathcal X^{IC,\mathcal C}$ if, for every signal profile $s$ and value profile $v$, we have $\sum_{i=1}^n x_i(v,s)\leq1$ and $0\leq x_i(v,s)\leq1$ for every $i\in N$, and, for each bidder $i$, every $s$, and every $v_{-i}$, we have that $x_i(v_i,v_{-i},s)$ is nondecreasing in $v_i$. Any allocation in $\mathcal X^{IC,\mathcal C}$ is implementable, conditional on the public signal profile $s$, by the usual envelope payment formula \eqref{eq:envelope}, with the lowest type receiving zero utility for each signal profile.

For an allocation-bonus pair $(x,q)$ under certification technology $\mathcal C$, define the execution incentive generated by the mechanism as
\[
I^{\mathcal C}(x,q) \coloneqq  \E\left[ \sum_{i=1}^n \Delta_i^{\mathcal C}(v_i,s_i)q_i(v,s) \right],
\]
where the expectation is taken over certified signals and values. Feasible bonus expenditures satisfy $0\leq q_i(v,s)\leq B_i x_i(v,s)$ for every $i$, $v$, and $s$. The certification problem is strictly feasible if
\begin{equation}
\exists\,(\tilde x,\tilde q) \quad \text{such that} \quad \tilde x\in\mathcal X^{IC,\mathcal C},\quad 0\leq \tilde q_i(v,s)\leq B_i\tilde x_i(v,s), \quad\text{and}\quad I^{\mathcal C}(\tilde x,\tilde q)>\psi.
\label{eq:certification_slater}
\end{equation}
This condition is the certified-signal analogue of the strict-feasibility condition used in the baseline problem. It guarantees that the execution agent's incentive constraint admits a supporting multiplier after the certification technology has fixed the public signal structure.

The previous results show why private reliability cannot generally be treated as a public quality characteristic. Certification changes the problem by turning some execution-relevant information into a public signal before the auction. Once the signal is public, the designer can condition both the allocation and the execution contract on it. The relevant execution parameters are then conditional on the certified signal, and the optimal mechanism becomes a certified shadow-allocation problem.

\begin{proposition}
\label{prop:optimal_certification_given_signal}
Fix a certification technology $\mathcal C$, and suppose \eqref{eq:certification_slater} holds. Then, there exists an optimal effort-inducing mechanism and a common ex ante multiplier $\lambda^{\mathcal C}\geq0$ such that the signal-contingent allocation rule solves
\[
x^{\mathcal C} \in \argmax_{x\in\mathcal X^{IC,\mathcal C}} \E\left[ \sum_{i=1}^n x_i(v,s) S_i^{\mathcal C}\left(v_i,s_i\mid\lambda^{\mathcal C}\right) \right],
\]
where
\[
S_i^{\mathcal C}\left(v_i,s_i\mid\lambda^{\mathcal C}\right) \coloneqq  \varphi_i(v_i\mid s_i) + g_i^{\mathcal C}(v_i,s_i) + B_i\max\left\{ \lambda^{\mathcal C}\Delta_i^{\mathcal C}(v_i,s_i) - \sigma_i^{\mathcal C}(v_i,s_i), 0\right\}.
\]
After the common multiplier $\lambda^{\mathcal C}$ is fixed, the mechanism chooses the best allocation conditional on each realized public signal profile, subject to monotonicity in values for that signal profile. The associated bonus expenditure satisfies
\[
q_i^{\mathcal C}(v,s) \in \argmax_{0\leq q\leq B_i x_i^{\mathcal C}(v,s)} \left( \lambda^{\mathcal C}\Delta_i^{\mathcal C}(v_i,s_i) - \sigma_i^{\mathcal C}(v_i,s_i) \right)q.
\]
The allocation is implemented by the envelope payment formula conditional on the public signal.
\end{proposition}

The multiplier $\lambda^C$ is common across signal profiles because the executor's incentive constraint is ex ante. Once this multiplier is fixed, the allocation and bonus rule may vary with the realized public signal profile.

Let $V(\mathcal C)$ denote the gross value of the optimal effort-inducing mechanism after using certification technology $\mathcal C$, before subtracting the certification cost. Say that $\mathcal C'$ is more informative than $\mathcal C$ if the signal generated by $\mathcal C$ can be obtained from the signal generated by $\mathcal C'$ through a public garbling. Thus, after observing the finer signal, the seller can commit to use only the coarser information generated by the garbling.

Certification affects the auction by determining which execution-relevant information becomes scoreable. A more informative certification technology allows the seller to condition the allocation and the execution contract on a finer public signal. Since the seller can always ignore extra information, a more informative technology cannot lower the gross value of the optimal effort-inducing mechanism. Of course, certification may be costly, so the seller trades off the gross value of better information against the cost of producing it. The next corollary formalizes this standard logic.

\begin{corollary}
\label{prop:value_of_certification}
If $\mathcal C'$ is more informative than $\mathcal C$, then $V(\mathcal C')\geq V(\mathcal C)$. The seller chooses a certification technology solving $\max_{\mathcal C} \left\{ V(\mathcal C)-K(\mathcal C) \right\}$. No certification gives the posterior shadow allocation based only on value reports. Full certification of reliability gives the public reliability shadow allocation.
\end{corollary}

Proposition~\ref{prop:optimal_certification_given_signal} and Corollary~\ref{prop:value_of_certification} give a formal meaning to scoreability. Certification determines which execution-relevant characteristics become public inputs to the shadow allocation. Thus, a scoring rule may require an informational instrument, such as auditing, pre-qualification, or disclosure, before the relevant execution characteristic can be used in the auction.

The analysis assumes that the certification signal is public. Appendix~\ref{app:private_certification} shows that public disclosure is not essential when the certified signal is relevant for execution but not informative about values. In that case, a signal that is private to the seller can be used to condition the execution contract without changing the bidder-side virtual-value formula.


\subsection{Execution externalities}
\label{sec:externalities}

The shadow-score theorem assumes that the execution parameters associated with bidder $i$ depend on bidder $i$'s allocation state, and not on strategic reports by losing bidders. In some applications this restriction may fail. The ease of execution may depend on the full auction environment. Losing bidders may create litigation risk, political opposition, or information that helps the executor monitor the winner. The gap between bids may affect the credibility of the winner. A regulator's effort may be more productive when the winner belongs to one group and the strongest loser belongs to another. These effects are execution externalities.

If execution externalities depend only on public characteristics of the bidder pool, the preceding characterization is unchanged after conditioning on the public state. The shadow term remains independent of bidder $i$'s own report except through $\varphi_i(v_i)$, so monotonicity is preserved under regularity.

The difficulty arises when the execution value of assigning the object to one bidder depends on another bidder's report. Then, bidder $i$'s report may affect not only bidder $i$'s own score, but also the scores of other bidders. Pointwise maximization of shadow scores can violate monotonicity. The optimal mechanism can still be characterized, but the characterization is no longer a bidder-specific scoring rule. It is a monotonicity-constrained shadow allocation.

Suppose that the execution parameters are public, but they may depend on the full report profile. If bidder $i$ receives the object at profile $v$, let $g_i(v)$, $\Delta_i(v)$, $\sigma_i(v)$, and $B_i(v)$ denote the corresponding execution payoff, marginal effect of effort, success probability under effort, and pledgeable success bonus, respectively.\footnote{These functions are primitives of the execution environment, evaluated at the report profile used by the auction. Thus, reports affect not only who wins, but also the execution problem created by assigning the object to a given bidder. This is precisely why incentive compatibility is more delicate. One bidder's report may change the execution value of assigning the object to another bidder.} Assume these functions are bounded and measurable, with $B_i(v)\geq0$. For an allocation-bonus pair $(x,q)$, feasible bonus expenditures satisfy $0\leq q_i(v)\leq B_i(v)x_i(v)$, and the execution incentive generated by $(x,q)$ is
\[
I^E(x,q) \coloneqq  \E\left[ \sum_{i=1}^n \Delta_i(v)q_i(v) \right].
\]
As usual, assume strict feasibility, i.e.,
\begin{equation}
\exists\,(\tilde x,\tilde q) \quad \text{such that} \quad \tilde x\in\mathcal X^{IC},\quad 0\leq\tilde q_i(v)\leq B_i(v)\tilde x_i(v), \quad\text{and}\quad I^E(\tilde x,\tilde q)>\psi.
\label{eq:externality_slater}
\end{equation}
For each bidder $i$, define the unconstrained externality-adjusted value
\[
\chi_i^E\left(v\mid\lambda^E\right) \coloneqq  B_i(v)\max\left\{\lambda^E\Delta_i(v)-\sigma_i(v),0\right\}, \] and the unconstrained externality-adjusted payoff \[ \mathcal{E}_i\left(v\mid\lambda^E\right) \coloneqq  \varphi_i(v_i)+g_i(v)+\chi_i^E\left(v\mid\lambda^E\right).
\]

The preceding discussion shows why execution externalities create a difficulty that is absent from the baseline model. When the execution value of assigning the object to bidder $i$ depends on the full report profile, the Lagrangian still has a shadow-payoff representation, but this payoff is no longer a bidder-specific score depending only on $i$'s own report. The next proposition characterizes the optimal mechanism as the corresponding shadow-allocation problem with monotonicity imposed directly.


\begin{proposition}
\label{prop:optimal_externalities}
Suppose execution primitives may depend on the full report profile, and suppose \eqref{eq:externality_slater} holds. Then, there exists an optimal effort-inducing mechanism and a multiplier $\lambda^E\geq0$ where the optimal allocation solves
\[
x^E \in \argmax_{x\in\mathcal X^{IC}} \E\left[ \sum_{i=1}^n x_i(v)\mathcal{E}_i\left(v\mid\lambda^E\right) \right].
\]
The associated bonus expenditure satisfies
\[
q_i^E(v) \in \argmax_{0\leq q\leq B_i(v)x_i^E(v)} \left(\lambda^E\Delta_i(v)-\sigma_i(v)\right)q.
\]
The allocation is implementable in dominant strategies by the envelope payment formula. If the unconstrained maximizer of $\mathcal{E}_i\left(v\mid\lambda^E\right)$ belongs to $\mathcal X^{IC}$, then it is optimal. If it does not, the optimal auction is the monotonicity-constrained shadow allocation.
\end{proposition}

Proposition~\ref{prop:optimal_externalities} is the externality analogue of the shadow-score theorem. The difference is that the payoff from assigning the object to bidder $i$ may depend on the reports of other bidders. The Lagrangian objective cannot generally be written as a collection of bidder-specific scores depending only on each bidder's own report and public characteristics. The optimal auction is a constrained allocation rule chosen to maximize the externality-adjusted shadow objective subject to incentive compatibility.

The following example shows why the monotonicity constraint can bind.

\begin{example}
\label{ex:externality_nonmonotone}
There are two bidders with values independently and uniformly distributed on $[0,1]$, so $\varphi_i(v_i)=2v_i-1$. This example considers the Lagrangian rule for a fixed multiplier $\lambda$. The purpose is to show that, with execution externalities, maximization of externality-adjusted shadow payoffs can violate monotonicity.

Fix $K>1/5$ and set $\lambda=1+2K$. Execution primitives depend on the report profile. If bidder 1 receives the object, let $g_1(v)=0$, $B_1(v)=1$, and
\[
\Delta_1(v)=\sigma_1(v) = \frac12 \mathds 1\{v_2\geq 0.6\}.
\]
Bidder 2's report affects how useful assigning the object to bidder 1 is for providing execution incentives. When $v_2<0.6$, bidder 1's allocation state provides no execution incentive. When $v_2\geq0.6$, effort raises the success probability by $1/2$, and success occurs with probability $1/2$ under effort. If bidder 2 receives the object, let $g_2(v)=B_2(v)=\Delta_2(v)=\sigma_2(v)=0$.

For the fixed multiplier $\lambda=1+2K$, the externality-adjusted shadow payoff from assigning the object to bidder 1 is
\[
\mathcal E_1(v\mid\lambda) = 2v_1-1 + \underbrace{B_1(v)\max\{\lambda\Delta_1(v)-\sigma_1(v),0\}}_{=K\mathds 1\{v_2\geq0.6\}}.
\]
The shadow payoff from assigning the object to bidder 2 is
\[
\mathcal E_2(v\mid\lambda)=2v_2-1.
\]

Now fix bidder 1's value at $v_1=0.51$. If bidder 2 reports $v_2=0.55$, then $\mathcal E_1=0.02<\mathcal E_2=0.10$. The shadow rule assigns the object to bidder 2. If bidder 2 instead reports $v_2=0.61$, then $\mathcal E_1=0.02+K>0.22=\mathcal E_2$, where the inequality uses $K>1/5$. The shadow rule now assigns the object to bidder 1.

Thus, bidder 2 wins at $v_2=0.55$ but loses at the higher report $v_2=0.61$, holding $v_1$ fixed. Bidder 2's allocation probability is not monotone in its own report. The shadow rule cannot be implemented by a truthful auction with critical value payments.
\end{example}

The above is a counterexample to pointwise shadow maximization in the externalities environment. The score adjustment is generated by execution externalities, as bidder 2's report changes the execution value of assigning the object to bidder 1. The resulting rule is not a truthful mechanism, because it violates monotonicity. This is why Proposition~\ref{prop:optimal_externalities} characterizes the optimal auction as a maximization of the externality-adjusted shadow objective over $\mathcal X^{IC}$, rather than as an assignment to the highest externality-adjusted shadow payoff.

The example does more than show that the shadow rule is not truthfully implementable. It also rules out the standard bidder-specific scoring interpretation. In any scoring rule with bidder-specific scores that are weakly increasing in own value, a bidder who strictly wins at some report cannot strictly lose after raising its own report, holding the other reports fixed. Example~\ref{ex:externality_nonmonotone} violates exactly this property.

The next lemma identifies the representation failure. A standard bidder-specific score rule has an upper-set property in each bidder's own report. Execution externalities can violate this property because one bidder's report may raise the shadow value of assigning the object to another bidder. The optimal truthful mechanism must impose monotonicity directly, as in Proposition~\ref{prop:optimal_externalities}.

\begin{lemma}
\label{prop:no_standard_scoring_representation}
Consider any allocation rule that can be represented by bidder-specific scores $s_i(v_i)$, where each $s_i$ is weakly increasing in $v_i$, and where the object is assigned to a bidder with the highest nonnegative score. Fix bidder $i$ and $v_{-i}$. If bidder $i$ is the unique winner at $v_i$, then bidder $i$ remains a maximal-score bidder at every $v_i'>v_i$. In particular, if bidder $i$ strictly beats the outside option and all other bidders at $v_i$, then bidder $i$ cannot strictly lose at any higher report $v_i'>v_i$. The execution-externality rule in Example~\ref{ex:externality_nonmonotone} has no such representation.
\end{lemma}

The externality problem arises because a bidder's report can affect the execution value of assigning the object to someone else. A related but distinct departure occurs when execution depends on the composition of a winner set. Even with public primitives and no report externalities, the shadow adjustment is not generally bidder-specific. The next Section~\ref{subsec:set_dependent_execution} studies this case.


\subsection{Set-dependent execution}
\label{subsec:set_dependent_execution}

The baseline model has a single object, so winner identity is the only allocation state that matters for execution. Many policy auctions allocate several licenses, concessions, regions, or access rights. In those settings, execution may depend on the set of winners. Monitoring may be easier when winners use compatible technologies, harder when several high-risk firms win together, or more costly when projects are geographically dispersed. These effects cannot generally be represented by bidder-specific quality terms.

Suppose there are up to $K$ identical objects. Each bidder wants at most one object, and has private value $v_i$. Let
\[
\mathcal A_K\coloneqq \{A\subseteq N \;\text{ such that }\; |A|\leq K\}
\]
be the set of feasible winner sets. The empty set $\varnothing\in\mathcal A_K$ represents no allocation. An allocation rule is a collection $\{y_A(v)\}_{A\in\mathcal A_K}$, where $y_A(v)\in[0,1]$ is the probability that the winner set is $A$, and $\sum_{A\in\mathcal A_K}y_A(v)=1$ for every $v$. Bidder $i$'s allocation probability is
\[
x_i(v)\coloneqq \sum_{A\ni i}y_A(v).
\]

If winner set $A$ is selected, the execution parameters are $g_A$, $\Delta_A$, $\sigma_A$, and $B_A$. For the empty set, normalize $g_\emptyset=\Delta_\emptyset=\sigma_\emptyset=B_\emptyset=0$. Thus, choosing $\emptyset$ gives the seller the outside option of not allocating any object, and creates no execution incentive or execution reward. The executor can receive a success-contingent bonus after the winner set $A$ is selected, with cap $B_A$. Let
\[
q_A(v)\coloneqq y_A(v)b_A(v)
\]
be promised bonus expenditure for winner set $A$. Feasibility requires $0\leq q_A(v)\leq B_A y_A(v)$. Conditional on inducing effort, the seller's expected payoff is
\[
\E\left[ \sum_{i=1}^n p_i(v) + \sum_{A\in\mathcal A_K}y_A(v)g_A - \sum_{A\in\mathcal A_K}\sigma_A q_A(v) \right],
\]
and the execution incentive constraint is
\[
\E\left[ \sum_{A\in\mathcal A_K}\Delta_A q_A(v) \right]\geq \psi.
\]

For $\lambda\geq0$, define the set-level shadow value $\chi_A(\lambda)$ as
\[
\chi_A(\lambda) \coloneqq  B_A\max\{\lambda\Delta_A-\sigma_A,0\},
\]
and the set-level execution adjustment
\[
H_A(\lambda)\coloneqq g_A+\chi_A(\lambda),
\]
with the convention that $H_\emptyset(\lambda^A)=0$.

Section~\ref{sec:externalities} considers execution externalities in a single-object environment. The allocation still selected one winner, but the execution value of assigning the object to one bidder could depend on other bidders' reports. This extension considers a different departure from bidder separability, where execution may depend on the entire winner set.

In this environment, the execution problem no longer attaches a shadow value to each bidder separately. Instead, each feasible winner set $A$ generates its own execution payoff and its own contribution to the execution agent's incentive constraint. For a candidate multiplier $\lambda$, the shadow value of choosing set $A$ is $H_A(\lambda)$. The virtual objective from choosing set $A$ is
\[
\sum_{i\in A}\varphi_i(v_i)+H_A(\lambda).
\]
The natural analogue of the baseline shadow-score rule is not an individual scoring rule, but a set-level shadow allocation rule. This procedure compares all feasible winner sets, including the empty set. Since the empty set has normalized shadow surplus zero, the mechanism chooses no allocation whenever every nonempty winner set gives negative shadow surplus.

I further impose the corresponding strict-feasibility condition. There must exist a feasible set-allocation rule $\{\tilde y_A\}_{A\in\mathcal A_K}$ such that the induced bidder allocation probabilities $\tilde x_i(v)\coloneqq \sum_{A\ni i}\tilde y_A(v)$ are nondecreasing in $v_i$, and
\begin{equation}\label{eq:set_slater}
\mathbb E\left[ \sum_{A\in\mathcal A_K}\Delta_A B_A \tilde y_A(v) \right]>\psi .
\end{equation}
This condition says that some incentive compatible set allocation can generate strictly more than the required execution incentive when maximal success-contingent bonuses are used.

The next proposition characterizes the optimal effort-inducing mechanism in this set-dependent environment. Under regularity, the set-level rule is monotone in each bidder's value. Increasing $v_i$ raises the objective of every set containing bidder $i$ and leaves the objective of every set not containing bidder $i$ unchanged. With an appropriate monotone tie-breaking rule, the allocation is implementable by critical value payments.\footnote{ Assumption~\ref{ass:regularity} is needed here because the proposition gives a pointwise set-level allocation rule rather than a maximization over an incentive compatible allocation set. The proof uses that raising $v_i$ raises $\varphi_i(v_i)$, and thus raises the objective of every winner set containing bidder $i$. }

\begin{proposition}
\label{prop:set_dependent_shadow_allocation}
Suppose that Assumption~\ref{ass:regularity} and \eqref{eq:set_slater} hold. Then, there exists an optimal effort-inducing mechanism and a multiplier $\lambda^A\geq0$ such that the allocation assigns positive probability only to winner sets that maximize the nonnegative set-level shadow surplus, i.e.,
\[
y_A^\ast(v)>0 \implies A \in \argmax_{C\in\mathcal A_K} \left\{ \sum_{i\in C}\varphi_i(v_i)+H_C\left(\lambda^A\right) \right\}.
\]
The associated bonus expenditure satisfies
\[
q_A^\ast(v) \in \argmax_{0\leq q\leq B_A y_A^\ast(v)} \left(\lambda^A\Delta_A-\sigma_A\right)q.
\]
The induced bidder allocation probabilities $x_i^\ast(v)$ are monotone in each bidder's own value and are implementable in dominant strategies by the envelope payment formula.
\end{proposition}

Proposition~\ref{prop:set_dependent_shadow_allocation} generalizes the shadow-score theorem from bidder-specific execution to set-dependent execution. The optimal allocation is still a shadow allocation, but the shadow adjustment is attached to a winner set $A$, not to individual bidders. The mechanism ranks feasible sets by virtual surplus plus the shadow value of implementing that set.

The set-level rule reduces to individual scoring only when the set-level execution adjustment is additive across winners. A set function $H(\cdot\mid \lambda)$ is modular if there exist numbers $\{\eta_i(\lambda)\}_{i\in N}$ such that
\[
H_A(\lambda)=\sum_{i\in A}\eta_i(\lambda)
\]
for every feasible winner set $A$. In that case, the set-level objective can be written as a sum of individual scores. If $H_A(\lambda)$ is not modular, the shadow value of including bidder $i$ depends on the other selected bidders, and the allocation must be represented as a set-level shadow allocation.

The modularity condition is a statement about whether the set-level execution adjustment can be decomposed into bidder-specific score adjustments. By a standard individual-score representation, I mean a rule that chooses a feasible winner set by maximizing
\[
\sum_{i\in A}\left(\varphi_i(v_i)+\eta_i(\lambda)\right)
\]
for some bidder-specific constants $\eta_i(\lambda)$. This is the natural multi-object analogue of the baseline shadow-score auction, where each bidder receives its own score adjustment. If $H_A(\lambda)$ is modular, then the set-level objective has this form. If it is not modular, the execution adjustment is intrinsically set-level rather than bidder by bidder.

\begin{corollary}
\label{prop:nonmodular_no_individual_scoring}
If $H_A(\lambda^A)$ is modular, then the set-dependent shadow allocation can be written in standard individual-score form, with bidder scores $\varphi_i(v_i)+\eta_i(\lambda^A)$.

If $H_A(\lambda^A)$ is not modular, then the set-dependent shadow objective cannot be written in standard individual-score form. In that case, the execution adjustment cannot generally be reduced to bidder-specific score premia added to virtual values. The allocation must instead compare feasible winner sets.
\end{corollary}

Nonmodularity is a clean case in which individual scoring fails. The execution object is the winner set, and not the individual bidder. If execution exhibits complementarities or substitution effects across winners, the mechanism must compare feasible sets rather than assign each bidder an independent score adjustment. This is relevant for spectrum packages, regional concessions, procurement lots, and platform access decisions in which the cost of execution depends on the composition of the selected firms.

\section{Concluding remarks}
\label{sec:conclusion}

This paper studies auctions in which allocation creates a moral hazard problem. A third-party executor must exert noncontractible effort before the winner is known, and the productivity and cost of rewarding that effort depend on who wins. In a regular IPV environment, the optimal effort-inducing mechanism is a shadow-score auction. Bidders are ranked by virtual values plus an execution adjustment. The novel component of this adjustment is the endogenous shadow value of relaxing the executor's incentive constraint. Assigning the object to some bidders is more valuable because those allocation states make execution effort easier to induce.

The result changes the usual interpretation of reserve-price design. In the standard Myersonian environment \citep{Myerson1981}, the reserve price determines whether the object is sold, and, conditional on sale, the object goes to the bidder with the highest relevant value. With winner-dependent incentives, reserve prices are not enough. A common reserve changes only the sale threshold. Bidder-specific reserves change eligibility. By contrast, shadow scores can also change the ranking of eligible bidders. This ranking channel is essential when the winner identity affects the cost of inducing execution effort.

This distinction matters for policy and regulation. In applications such as spectrum licenses, infrastructure concessions, public-private partnerships, privatizations, and digital platform access, the highest monetary bidder need not be the one who makes execution easiest or most valuable. A regulator may need to monitor coverage obligations, enforce investment milestones, approve restructuring plans, or supervise compliance after the auction. If that effort is more productive or less costly for some winners than for others, then a revenue-only auction may allocate the object to a bidder in a way that is difficult to implement. The shadow-score auction internalizes this execution effect in the allocation rule.

The model also highlights a transparency issue. A shadow-score auction may award the object to a bidder with a lower monetary bid because that bidder makes it easier to provide execution incentives. In regulatory settings, such departures from price ranking can be controversial. The analysis suggests that execution-based scoring rules should be announced, justified, and tied to verifiable parameters. The score should not be an opaque discretionary preference for one bidder over another, but should represent an execution value that the designer can explain and, where possible, document.

The analysis also clarifies what a regulator or public buyer must know to use such a rule. If execution parameters are public, audited, or certifiable, then they can be entered directly into the shadow score. If reliability is private, the public information score may be unattainable. When reliability affects bidders only through willingness to pay, the mechanism can use the posterior execution values associated with that willingness to pay. However, residual reliability remains bunched unless it is verified. Certification, pre-qualification, auditing, disclosure rules, and performance guarantees determine which execution characteristics become scoreable. These instruments shape the allocation rules that can be truthfully implemented.

The paper also identifies the limits of the scoring interpretation. With public winner-specific execution parameters, the optimal allocation has a reduced-form scoring representation. Once execution information is private, report-dependent, or set-dependent, that representation can fail. The optimal mechanism then becomes a posterior, externality-adjusted, or set-level shadow allocation subject to incentive compatibility.

The broader contribution is that the auction does more than screen bidders and extract revenue, as it also selects the problem faced by the executor. Once execution is part of the design problem, the standard logic of reserve auctions and the quality-preference interpretation of scoring auctions are both incomplete. The optimal auction may need to allocate not only to the bidder with the highest willingness to pay, but also to the bidder or winner set that makes the allocation easiest to implement.



\newpage

\appendix


\section{Proofs}


\subsection{Baseline analysis}

\noindent{\bf Proof of Lemma~\ref{prop:fixed}.} By the revenue equivalence identity, for any incentive compatible and individually rational mechanism, expected bidder payments are at most expected virtual surplus, with equality when the lowest type of each bidder obtains zero utility. Hence, conditional on effort and the fixed bonus $B$, the seller maximizes
\[
\E\left[\sum_i x_i(v)\left(\varphi_i(v_i)+g_i-\sigma_iB\right)\right]
\]
subject to feasibility and \eqref{eq:fixedIC}. The Lagrangian for \eqref{eq:fixedIC} is
\[
\E\left[\sum_i x_i(v)\left(\varphi_i(v_i)+g_i-\sigma_iB+\lambda B\D_i\right)\right]-\lambda\psi.
\]
For a fixed $\lambda$, this expression is maximized pointwise by assigning the object to the bidder with the highest nonnegative score \eqref{eq:fixedscore}. Since $\varphi_i$ is strictly increasing and all other terms in $S_i$ are independent of $v_i$, bidder $i$'s allocation is nondecreasing in its report for every $v_{-i}$. As a result, the allocation is dominant-strategy implementable by the critical value payment rule. Complementary slackness gives the optimal multiplier. \qed 

\bigskip

\noindent{\bf Proof of Lemma~\ref{lem:effort_feasibility}.}
For any feasible allocation-bonus pair $(x,q)$,
\[
\sum_{i=1}^n \D_i q_i(v) \leq \sum_{i=1}^n \D_i B_i x_i(v) \leq \bar I\sum_{i=1}^n x_i(v) \leq \bar I
\]
for every value profile $v$. Taking expectations gives $I(x,q)\leq \bar I$. It follows that effort cannot be induced if $\psi>\bar I$.

Conversely, let $k\in\argmax_i \D_i B_i$. The mechanism that assigns the object to bidder $k$ for every value profile, and sets $q_k(v)=B_k$, generates execution incentive $I(x,q)=\D_k B_k=\bar I$. Hence, effort is feasible whenever $\psi\leq\bar I$, and strictly feasible whenever $\psi<\bar I$. \qed

\bigskip

\noindent{\bf Proof of Lemma~\ref{lem:multiplier}.} For each required incentive level $z$, define the value function
\[
V(z) \coloneqq \sup_{x,q} \left\{ \Pi(x,q) \; \text{ such that } \; \sum_{i=1}^n x_i(v)\leq1,\; 0\leq q_i(v)\leq B_i x_i(v),\; I(x,q)\geq z \right\}.
\] The function $V$ is finite because values and bonuses are bounded. It is also concave, as randomizing between two feasible allocation-bonus pairs randomizes both payoff and the generated execution incentive. Moreover, it is nonincreasing in $z$, since a higher required incentive level tightens the constraint.

By \eqref{eq:slater}, $\psi$ lies in the relative interior of the feasible range of incentive levels.\footnote{This is the usual Slater-type strict feasibility condition for the relaxed linear problem. It ensures that the executor's incentive constraint admits a supporting Lagrange multiplier. See, for example, \citet{Sundaram1996}.} Since $V$ is finite and concave on this interval, there exists a supporting slope $s$ at $\psi$ such that $V(z)\leq V(\psi)+s(z-\psi)$ for every feasible incentive level $z$. Since $V$ is nonincreasing, $s\leq0$. Define $\lambda^\ast\coloneqq -s\geq0$. Then, for every incentive level $z$,
\[
V(z)\leq V(\psi)-\lambda^\ast(z-\psi).
\]
For any feasible allocation-bonus pair $(x,q)$, applying this inequality at $z=I(x,q)$ gives
\[
\Pi(x,q)+\lambda^\ast\left(I(x,q)-\psi\right)\leq V(\psi).
\]
Hence, the Lagrangian with multiplier $\lambda^\ast$ is maximized by any primal optimizer of the effort-inducing problem. Complementary slackness follows immediately. If $(x^\ast,q^\ast)$ attains $V(\psi)$, then $\lambda^\ast\left(I(x^\ast,q^\ast)-\psi\right)=0$. \qed


\bigskip

\noindent{\bf Proof of Theorem~\ref{thm:execution_score}.} By the standard revenue identity for single-dimensional private values, any incentive compatible and individually rational mechanism satisfies
\[
\E\left[\sum_{i=1}^n p_i(v)\right] \leq \E\left[\sum_{i=1}^n \varphi_i(v_i)x_i(v)\right],
\]
with equality when the lowest type of each bidder obtains zero expected utility. Conditional on inducing effort, the seller's payoff is bounded above by
\[
\E\left[ \sum_{i=1}^n x_i(v)\left(\varphi_i(v_i)+g_i\right) - \sum_{i=1}^n \sigma_i q_i(v) \right].
\]

The relaxed effort-inducing problem is 
\[
\max_{x,q} \E\left[ \sum_{i=1}^n x_i(v)\left(\varphi_i(v_i)+g_i\right) - \sum_{i=1}^n \sigma_i q_i(v) \right]
\]
subject to
\[
\sum_{i=1}^n x_i(v)\leq1,
\]
\[
0\leq x_i(v)\leq1,
\]
\[
0\leq q_i(v)\leq B_i x_i(v),
\]
and
\[
\E\left[\sum_{i=1}^n \D_i  q_i(v)\right]\geq\psi.
\]
This is a linear program in the variables $(x,q)$.

Attach a multiplier $\lambda\geq0$ to the execution incentive constraint. The Lagrangian is
\[
\E\left[ \sum_{i=1}^n x_i(v)\left(\varphi_i(v_i)+g_i\right) + \sum_{i=1}^n (\lambda\D_i -\sigma_i)q_i(v) \right] -\lambda\psi.
\]
For fixed $x_i(v)$, the constraint $0\leq q_i(v)\leq B_i x_i(v)$ implies that the optimal $q_i(v)$ solves
\[
\max_{0\leq q\leq B_i x_i(v)} (\lambda\D_i -\sigma_i)q.
\]
Alternatively, whenever bidder $i$ is assigned the object, the corresponding bonus solves
\[
\max_{b\in[0,B_i]}(\lambda\D_i -\sigma_i)b.
\]
The value per unit of allocation is
\[
\chi_i(\lambda) = \max_{b\in[0,B_i]}(\lambda\D_i -\sigma_i)b.
\]
Substituting this maximized value into the Lagrangian gives the objective
\[
\sum_{i=1}^n x_i(v) \left(\varphi_i(v_i)+g_i+\chi_i(\lambda)\right) = \sum_{i=1}^n x_i(v)S_i(v_i\mid \lambda).
\]
For any fixed $\lambda$, the Lagrangian is maximized pointwise by assigning the object to a bidder with the highest nonnegative score.

By Lemma~\ref{lem:multiplier}, strict feasibility and attainment give a supporting multiplier $\lambda^\ast \geq 0$ for an optimal effort-inducing allocation-bonus pair $(x^\ast,q^\ast)$. Complementary slackness gives
\[
\lambda^\ast \left( \E\left[\sum_{i=1}^n \D_i q_i^\ast(v)\right]-\psi \right) =0.
\]
Away from the knife-edge case $\lambda^\ast\Delta_i=\sigma_i$, the optimal bonus is bang-bang. Whenever bidder $i$ wins, the bonus is either zero or the maximal
pledgeable bonus $B_i$. This proves part~(i).

Second, we show implementability. For fixed $\lambda^\ast$, bidder $i$'s score is $S_i(v_i\mid \lambda^\ast)=\varphi_i(v_i)+g_i+\chi_i(\lambda^\ast)$. By Assumption~\ref{ass:regularity}, $\varphi_i$ is strictly increasing. The term $(g_i+\chi_i(\lambda^\ast))$ is independent of bidder $i$'s report, and so $S_i(v_i\mid \lambda^\ast)$ is strictly increasing in $v_i$. Holding $v_{-i}$ fixed, bidder $i$'s score is strictly increasing in $v_i$. Hence, once bidder $i$ is selected by the shadow-score rule, a higher report cannot make it lose. The fixed priority rule resolves any score ties in a way that does not depend on bidder $i$'s report except through its score. Therefore, bidder $i$'s allocation rule is nondecreasing in its report.

The critical value payment rule stated in the theorem implements this monotone allocation in dominant strategies and gives the lowest type zero utility. Score ties occur only on boundary profiles under the atomless type distributions, so the particular tie-breaking rule has no effect on expected payoffs or revenue. This proves part (ii). \qed


\bigskip

\noindent{\bf Proof of Corollary~\ref{cor:execution_cost_index}.} If $\D_i>0$, then $\lambda\D_i-\sigma_i = \D_i(\lambda-\rho_i)$. Substituting this expression into
\[
\chi_i(\lambda)=B_i\max\{\lambda\D_i-\sigma_i,0\}
\]
gives the stated formula. The bonus characterization follows from the maximization problem
\[
\max_{b\in[0,B_i]}(\lambda^\ast\D_i-\sigma_i)b.
\]
If $\D_i=0$, then $\lambda^\ast\D_i-\sigma_i\leq0$, so the allocation state does not generate positive execution incentives. \qed


\bigskip

\noindent{\bf Proof of Corollary~\ref{prop:lambda_comparative_statics}.}
The dual representation follows from Lemma~\ref{lem:multiplier}. To prove monotonicity in $\psi$, let $\psi_2>\psi_1$, $\lambda_1\in\Lambda(\psi_1\mid B)$, and $\lambda_2\in\Lambda(\psi_2\mid B)$. Optimality gives
\[
H(\lambda_1\mid B)-\lambda_1\psi_1 \leq H(\lambda_2\mid B)-\lambda_2\psi_1
\]
and
\[
H(\lambda_2\mid B)-\lambda_2\psi_2 \leq H(\lambda_1\mid B)-\lambda_1\psi_2.
\]
Adding the two inequalities yields $(\lambda_2-\lambda_1)(\psi_2-\psi_1)\geq0$. Since $\psi_2>\psi_1$, this implies $\lambda_2\geq\lambda_1$.

For the local comparative statics, the first-order condition for an interior unique minimizer is $H_\lambda(\lambda^\ast(a)\mid a)=\psi$. Differentiating with respect to $\psi$ gives
\[
H_{\lambda\lambda}(\lambda^\ast(a)\mid a) \frac{d\lambda^\ast}{d\psi}=1,
\]
and differentiating with respect to $a$ gives
\[
H_{\lambda\lambda}(\lambda^\ast(a)\mid a) \frac{d\lambda^\ast}{da} + H_{\lambda a}(\lambda^\ast(a)\mid a)=0.
\]
The stated formulas follow. \qed


\bigskip

\noindent{\bf Proof of Corollary~\ref{prop:envelope_values}.} By strong duality, the effort-inducing value can be written as
\[
V^E(\psi,B) = \min_{\lambda\geq0} \left\{ H(\lambda\mid B)-\lambda\psi \right\},
\]
where $H(\lambda\mid B)$ is defined as in \eqref{eq:surplus}. The envelope theorem gives
\[
\frac{\partial V^E}{\partial \psi}=-\lambda^\ast.
\]
If the score maximizer is unique almost everywhere, differentiating the integrand of $H$ with respect to $B_i$ gives
\[
\frac{\partial V^E}{\partial B_i} = \E\left[ x_i^\ast(v)\max\{\lambda^\ast\D_i-\sigma_i,0\} \right].
\] 
\qed

\bigskip

\noindent{\bf Proof of Lemma~\ref{prop:post_win_effort}.}
If bidder $i$ wins, the seller faces a winner-specific execution problem. Inducing effort requires $\D_i b_i\geq\psi$ and is feasible if and only if $\D_i B_i\geq\psi$. When effort is feasible, the seller chooses the cheapest success-contingent bonus,
\[
b_i^E=\frac{\psi}{\D_i},
\]
which gives net execution payoff
\[
g_i-\sigma_i\frac{\psi}{\D_i}.
\]
If effort is not induced, the seller obtains $g_i^0$. It follows that the best continuation payoff after assigning the object to bidder $i$ is $\Gamma_i$.

The auction problem is the standard single-dimensional optimal auction problem with winner-specific payoff shifters $\Gamma_i$. By the revenue equivalence identity, expected payments are bounded above by expected virtual surplus. The seller maximizes
\[
\E\left[
\sum_{i=1}^n x_i(v)\left(\varphi_i(v_i)+\Gamma_i\right)
\right]
\]
subject to feasibility. Pointwise maximization assigns the object to a bidder with the highest nonnegative value of $\varphi_i(v_i)+\Gamma_i$. By Assumption~\ref{ass:regularity}, this score is strictly increasing in $v_i$, so the allocation is monotone in each bidder's report and is implemented by critical value payments. \qed

\bigskip

\noindent{\bf Proof of Proposition~\ref{prop:format_restrictions}.} For part~(i), suppose first that all bidders have the same shadow handicap $h^\ast$. Then, $S_i(v_i)=\varphi(v_i)+h^\ast$ for every bidder $i$. Since $\varphi$ is strictly increasing, the bidder with the highest shadow score is exactly the bidder with the highest value. The nonnegativity condition $S_i(v_i)\geq0$ is equivalent to $v_i\geq r^\ast$ and $\varphi(r^\ast)+h^\ast=0$, with the usual truncation if the solution lies outside the support. Hence, the shadow-score allocation is an anonymous second-price auction with common reserve $r^\ast$.

Conversely, suppose two active bidders $i$ and $j$ have different handicaps. Without loss of generality, let $h_i^\ast>h_j^\ast$. Because both bidders are active, each wins on a set of positive probability under the shadow-score allocation. Since $\varphi$ is continuous and strictly increasing on an interval, there is a positive measure set of value profiles on which $v_i<v_j$ but $\varphi(v_i)+h_i^\ast>\varphi(v_j)+h_j^\ast$, with both scores positive. On this set, the shadow-score allocation assigns the object to bidder $i$, even though bidder $j$ has the higher value. An anonymous second-price auction with a common reserve always assigns the object to the highest-value bidder among those above the reserve. It cannot implement this reversal.

For part~(ii), observe that bidder-specific reserves can alter eligibility but not value rankings among eligible bidders. If bidders $i$ and $j$ are both eligible and $v_j>v_i$, a bidder-specific-reserve auction assigns the object to bidder $j$, not bidder $i$.

Suppose two active bidders $i$ and $j$ have different handicaps, and again let $h_i^\ast>h_j^\ast$. Any bidder-specific-reserve execution must make both bidders eligible on nondegenerate upper tail intervals. Otherwise, one of them would not be active. These eligibility intervals overlap near the top of the common support. Choose $v_i$ and $v_j$ in this overlap with $v_i<v_j$ and $\varphi(v_i)+h_i^\ast>\varphi(v_j)+h_j^\ast$. Such pairs exist on a positive measure set because the handicap difference is strict and $\varphi$ is continuous. The shadow-score allocation gives the object to bidder $i$, whereas any bidder-specific-reserve auction gives it to bidder $j$, since both are eligible and $v_j>v_i$. Hence, no bidder-specific reserve auction implements the shadow-score allocation.

For part~(iii), let $y^\ast=(x^\ast,q^\ast)$ denote the optimal shadow-score mechanism. Since the execution constraint binds at $y^\ast$, we have $I(y^\ast)=\psi$. Consider the Lagrangian evaluated at the optimal multiplier $\lambda^\ast$,
\[
\mathcal L_{\lambda^\ast}(y) \coloneqq \Pi(y)+\lambda^\ast\left(I(y)-\psi\right).
\]
Using the definitions of $I(y)$ and $\Pi(y)$, we obtain
\[
\mathcal L_{\lambda^\ast}(y) = \E\left[ \sum_{i=1}^n x_i(v)\left(\varphi(v_i)+g_i\right) + \sum_{i=1}^n(\lambda^\ast\D_i-\sigma_i)q_i(v) \right] -\lambda^\ast\psi.
\]
For each bidder $i$ and profile $v$, the bonus-expenditure component is maximized subject to $0\leq q_i(v)\leq B_i x_i(v)$. The maximized value per unit of allocation is
\[
\chi_i(\lambda^\ast) = \max_{b\in[0,B_i]}(\lambda^\ast\D_i-\sigma_i)b.
\]
After substituting this value, the Lagrangian score of bidder $i$ is $\varphi(v_i)+g_i+\chi_i(\lambda^\ast)$. Therefore, $y^\ast$ maximizes $\mathcal L_{\lambda^\ast}$ over all
allocation-bonus pairs.

Now let $y=(x,q)$ be any effort-feasible restricted reserve mechanism. Since $I(y)\geq\psi$, we have that
\[
\mathcal L_{\lambda^\ast}(y) = \Pi(y)+\lambda^\ast\left(I(y)-\psi\right) \geq \Pi(y).
\]
Moreover, because $I(y^\ast)=\psi$, we obtain $\mathcal L_{\lambda^\ast}(y^\ast)=\Pi(y^\ast)$.

By assumption, no effort-feasible restricted reserve mechanism has allocation rule $x^\ast$ almost everywhere. Hence, $x$ differs from $x^\ast$ on a set of value profiles with positive probability. Since shadow-score maximization is unique almost everywhere, assigning the object according to $x$ rather than $x^\ast$ strictly lowers the maximized Lagrangian integrand on that set. If the restricted mechanism also uses bonus expenditures that do not maximize the Lagrangian bonus term, the Lagrangian is lowered further. As a result, $\mathcal L_{\lambda^\ast}(y^\ast)> \mathcal L_{\lambda^\ast}(y)$. Combining the preceding inequalities gives
\[
\Pi(y^\ast) = \mathcal L_{\lambda^\ast}(y^\ast) > \mathcal L_{\lambda^\ast}(y) = \Pi(y)+\lambda^\ast\left(I(y)-\psi\right) \geq \Pi(y).
\]
Every effort-feasible restricted reserve mechanism that fails to implement the shadow-score allocation yields strictly lower expected payoff than the shadow-score auction. \qed

\bigskip

\noindent{\bf Proof of Proposition~\ref{prop:common_reserve_loss}.} Part~(i) follows from the in-text computation.

Consider an anonymous second-price auction with common reserve $r$. Bidder 1 wins if and only if $v_1\geq r$ and $v_1\geq v_2$. Hence,
\[
\Pr(\text{bidder 1 wins}) = \int_r^1 v_1\,dv_1 = \frac{1-r^2}{2}.
\]
The execution constraint requires
\[
\frac{1-r^2}{2}\geq \kappa,
\]
or $r\leq \sqrt{1-2\kappa}$. Thus, common reserves are feasible if and only if $\kappa\leq1/2$.

Expected virtual surplus under common reserve $r$ is
\[
R^C(r) = \int_r^1 (2v-1)2v\,dv = \frac{1}{3}+r^2-\frac{4}{3}r^3.
\]
This expression is increasing on $[0,1/2]$, since
\[
\frac{dR^C(r)}{dr} = 2r(1-2r)\geq0
\]
for $r\in[0,1/2]$. For $\kappa\in(3/8,1/2]$, the constraint requires $r\leq \sqrt{1-2\kappa}<1/2$. Thus, the best feasible common reserve is $r^C(\kappa)=\sqrt{1-2\kappa}$. This proves part~(ii).

We now characterize the optimal shadow-score auction. Let $\lambda\geq0$ be the multiplier on the execution constraint \eqref{eq:kappa_constraint}. The Lagrangian scores are
\[
S_1(v_1\mid \lambda)=2v_1-1+\lambda,
\]
\[
S_2(v_2)=2v_2-1.
\]

Write $\delta(\lambda)=\lambda/2$. For $0\leq\delta(\lambda)\leq1/2$, bidder 1 wins if $v_1+\delta(\lambda)\geq v_2$ and $v_1\geq\frac{1}{2}-\delta(\lambda)$. Therefore,
\[
\Pr(\text{bidder 1 wins}) = \int_{1/2-\delta}^{1-\delta}(v_1+\delta(\lambda))\,dv_1 + \int_{1-\delta(\lambda)}^{1}1\,dv_1 = \frac{3}{8}+\delta(\lambda).
\]
To satisfy the constraint with equality, the optimal multiplier $\lambda^\ast(\kappa)$ satisfies
\[
\delta^\ast(\kappa)\coloneqq \frac{\lambda^\ast(\kappa)}{2} = \kappa-\frac{3}{8}.
\]
This formula applies for $\kappa\in\left[3/8,7/8\right]$, corresponding to $\delta\in[0,1/2]$.

The expected virtual surplus from the shadow-score auction is the expected virtual value of the winner. Bidder 1 wins on the region
\[
v_1\geq \frac{1}{2}-\delta^\ast(\kappa),
\]
\[
v_2\leq \min\{1,v_1+\delta^\ast(\kappa)\},
\]
and bidder 2 wins on the region
\[
v_2\geq \frac{1}{2},
\]
\[
v_2\geq v_1+\delta^\ast(\kappa).
\]
Hence,
\begin{multline}
R^S(\kappa) = \int_{1/2-\delta^\ast(\kappa)}^{1-\delta^\ast(\kappa)}(2v_1-1)(v_1+\delta^\ast(\kappa))\,dv_1 \\ + \int_{1-\delta^\ast(\kappa)}^{1}(2v_1-1)\,dv_1 + \int_{1/2}^{1}(2v_2-1)(v_2-\delta^\ast(\kappa))\,dv_2. 
\end{multline}
A direct calculation gives
\[
R^S(\kappa) = \frac{5}{12}-\delta^\ast(\kappa)^2 = \frac{5}{12} - \left(\kappa-\frac{3}{8}\right)^2.
\]
This proves part~(iii).

It remains to compare the shadow-score auction with the best common reserve. Let
\[
r\coloneqq r^C(\kappa)=\sqrt{1-2\kappa}.
\]
For $\kappa\in[3/8,1/2]$, we have $r\in[0,1/2]$ and
\[
\delta^\ast(\kappa) = \kappa-\frac{3}{8} = \frac{1-4r^2}{8}.
\]
Substituting into $R^S(\kappa)-R^C(\kappa)$ yields
\[
R^S(\kappa)-R^C(\kappa) = - \frac{(2r-1)^2\left(12r^2-52r-13\right)}{192}.
\]
For $r\in[0,1/2]$, the first factor is nonnegative and $12r^2-52r-13<0$. Thus, $R^S(\kappa)-R^C(\kappa)\geq0$, with equality only at $r=1/2$, equivalently $\kappa=3/8$. Hence, the shadow-score auction strictly dominates for $\kappa\in\left(3/8,1/2\right]$.

To find the maximum loss, write
\[
R^S(\kappa)-R^C(\kappa) = \frac{13}{192} -\frac{7}{8}r^2 +\frac{4}{3}r^3 -\frac{1}{4}r^4.
\]
Then, for $r\in[0,1/2]$,
\[
\frac{d\left(R^S(\kappa)-R^C(\kappa)\right)}{dr} = r\left(-\frac{7}{4}+4r-r^2\right)\leq0,
\]
with strict inequality for $r\in(0,1/2)$. The loss is maximized at $r=0$, or $\kappa=1/2$, where
\[
R^S(1/2)-R^C(1/2)=\frac{13}{192}.
\]
This proves part~(iv). \qed


\subsection{Beyond score auctions}\label{app:beyond}

\noindent{\bf Proof of Proposition~\ref{prop:private_bunching}.}
 Fix bidder $i$ and suppress it as a subscript. Fix two reliability types, $a$ and $a'$, and consider a value $v$ such that both $(v,a)$ and $(v,a')$ belong to the type support. For any report $(r,\alpha)$, write the interim utility of a type with value $v$ as
\[
U(v\mid r,\alpha)=vX(r,\alpha)-P(r,\alpha).
\]
Because reliability is payoff-irrelevant, this utility depends on the true type only through $v$.

Bayesian incentive compatibility gives
\[
vX(v,a)-P(v,a)\geq vX(v,a')-P(v,a')
\]
for type $(v,a)$, and
\[
vX(v,a')-P(v,a')\geq vX(v,a)-P(v,a)
\]
for type $(v,a')$. Hence,
\[
vX(v,a)-P(v,a)=vX(v,a')-P(v,a')
\]
at every value $v$ where both reliability types can occur. Let this common truthful utility be denoted by $U(v)$.

Now fix $a$. Incentive compatibility with respect to the value report implies the standard envelope formula,
\[
U(v)=U(\ubar v)+\int_{\ubar v}^{v}X(s,a)\,ds,
\]
at all points of differentiability of $U$. Applying the usual envelope argument along each reliability type, the derivative of this common truthful utility is $X(v,a)$ and also $X(v,a')$ for almost every $v$ in the common value support. Hence, $X(v,a)=X(v,a')$ almost everywhere on $V(a)\cap V(a')$. The equality of payments follows from the equality of truthful utilities. \qed

\bigskip

\noindent{\bf Proof of Corollary~\ref{cor:public_score_unattainable}.}
If a direct mechanism implemented the public reliability shadow-score allocation using unverifiable reports of $a_i$, its interim allocation probabilities would have to coincide with $X_i^{P}(v_i,a_i)$. But Proposition~\ref{prop:private_bunching} implies that any Bayesian incentive compatible mechanism must satisfy $X_i(v_i,a_i)=X_i(v_i,a_i')$ for almost every $v_i$. This contradicts the assumed positive measure set on which the public allocation gives different interim allocation probabilities to $a_i$ and $a_i'$. \qed

\bigskip

\noindent{\bf Proof of Proposition~\ref{prop:optimal_private_reliability}.}
Because reliability is payoff-irrelevant and unverifiable, Proposition~\ref{prop:private_bunching} implies that reliability reports cannot be used to generate different interim allocations or payments for types with the same value. The robust design problem conditions on value reports and uses posterior execution primitives.

Conditional on value reports, the seller's virtual payoff and the executor's incentive constraint are obtained from the baseline problem by replacing $(\Delta_i,\sigma_i,g_i,B_i)$ with the posterior primitives $\left(\bar\Delta_i(v_i),\bar\sigma_i(v_i),\bar g_i(v_i),\bar B_i\right)$. Attaching multiplier $\lambda$ to the posterior incentive constraint, the bonus-expenditure term for bidder $i$ is
\[
\max_{0\leq q\leq \bar B_i x_i(v)} \left(\lambda\bar\Delta_i(v_i)-\bar\sigma_i(v_i)\right)q,
\]
with maximized value
\[
x_i(v)\bar B_i \max\left\{\lambda\bar\Delta_i(v_i)-\bar\sigma_i(v_i),0\right\}.
\]
Thus, for a fixed multiplier, the allocation objective is
\[
\mathbb E\left[ \sum_{i=1}^n x_i(v)\bar S_i(v_i\mid\lambda) \right].
\]
The strict-feasibility condition \eqref{eq:private_slater} gives a supporting multiplier $\lambda^P$ by the same value function argument as Lemma~\ref{lem:multiplier}. Because $\bar S_i(v_i\mid\lambda^P)$ need not be increasing in $v_i$, the allocation must be chosen from $\mathcal X^{IC}$. The envelope payment formula then implements the resulting monotone allocation. \qed

\bigskip

\noindent{\bf Proof of Example~\ref{prop:posterior_nonmonotone}.}
Consider one bidder with value $v\sim U[0,1]$. Then, $\varphi(v)=2v-1$. Suppose that, for a fixed multiplier $\lambda$, the posterior execution term is
\[
\bar g(v)+\bar\chi(v\mid \lambda) =
\begin{cases}
0.8 & \text{ when } v\in[0,1/4],\\
0 & \text{ when } v\in(1/4,1].
\end{cases}
\]
This posterior term can be generated by a binary reliability type that is more likely at low values and that raises the execution value of allocation.

The posterior shadow score is
\[
\bar S(v\mid \lambda) =
\begin{cases}
2v-0.2 & \text{ when } v\in[0,1/4],\\
2v-1 & \text{ when } v\in(1/4,1].
\end{cases}
\]
The posterior-score rule allocates whenever $\bar S(v\mid \lambda)\geq0$. Hence, it allocates for $v\in[0.1,1/4]\cup[1/2,1]$. This allocation rule is not monotone in $v$. Therefore, it cannot be implemented by any incentive compatible mechanism.

We now show that the best monotone allocation discards the low value positive score interval. A general incentive compatible one bidder allocation rule is a nondecreasing function $\alpha:[0,1]\to[0,1]$, not necessarily a deterministic threshold. Since the objective is linear in $\alpha$, and any nondecreasing allocation rule can be represented as a mixture of threshold rules, an optimal monotone allocation can be chosen to be a threshold rule. Thus, it is enough to maximize
\[
J(r)\coloneqq \int_r^1 \bar S(v\mid \lambda)\,dv.
\]

For $r\geq1/4$,
\[
J(r)=\int_r^1(2v-1)\,dv=r(1-r),
\]
which is maximized at $r=1/2$ with value $1/4$. For $r\in[0,1/4]$,
\[
J(r) = \int_r^{1/4}(2v-0.2)\,dv + \int_{1/4}^{1}(2v-1)\,dv = 0.2r-r^2+0.2.
\]
This expression is maximized at $r=0.1$, where it equals $0.21<0.25$. Thus, the optimal monotone allocation is the threshold rule $r=1/2$.

The posterior-score maximizer would allocate to the interval $[0.1,1/4]$, but doing so would require excluding some higher values. The optimal incentive compatible mechanism discards that interval. This is the ironing problem, as posterior execution values may be economically valuable, but they cannot be used pointwise when doing so destroys monotonicity.

The posterior execution term used in the example can be generated from primitives consistent with Proposition~\ref{prop:optimal_private_reliability}. Let $a\in\{H,L\}$. Values are uniformly distributed on $[0,1]$, and reliability is correlated with value according to
\[
\Pr(a=H\mid v) =
\begin{cases}
0.8 & \text{ when }  v\in[0,1/4],\\
0 & \text{ when } v\in(1/4,1].
\end{cases}
\]
Let the bonus cap be $\bar B=1$, set $g(H)=g(L)=0$, and suppose that
\[
\Delta(H)=\sigma(H)=1 > 0=\Delta(L)=\sigma(L).
\]
Thus, the high-reliability type is one for which effort is productive and success occurs with probability one under effort, while the low-reliability type does not generate execution incentives. For the fixed multiplier $\lambda=2$, the posterior primitives are $\bar\Delta(v)=\bar\sigma(v)=\Pr(a=H\mid v)$. Hence, the posterior shadow term is
\[
\bar\chi(v\mid 2) = \max\left\{2\bar\Delta(v)-\bar\sigma(v),0\right\} = \Pr(a=H\mid v).
\]
Since $\bar g(v)=0$, this gives
\[
\bar g(v)+\bar\chi(v\mid 2) =
\begin{cases}
0.8 & \text{ when } v\in[0,1/4],\\
0 & \text{ when } v\in(1/4,1],
\end{cases}
\]
which is exactly the posterior execution term used in the example. \qed

\bigskip

\noindent{\bf Proof of Corollary~\ref{cor:wtp_bunching}.}
Conditional on $\theta_i$, reliability $a_i$ does not affect bidder $i$'s preferences over winning and paying. Therefore,Proposition~\ref{prop:private_bunching} applies with $\theta_i$ in place of $v_i$. Hence, all reliability types with the same willingness to pay must have the same interim allocation probability and the same interim payment. \qed

\bigskip

\noindent{\bf Proof of Corollary~\ref{prop:optimal_payoff_relevant_reliability}.}
By Corollary~\ref{cor:wtp_bunching}, any incentive compatible mechanism bunches reliability types with the same willingness to pay. Hence, the mechanism can condition only on $\theta_i$ and on posterior execution primitives conditional on $\theta_i$.

The rest of the argument is the payoff-relevant analogue of Proposition~\ref{prop:optimal_private_reliability}. The revenue identity is applied to the one-dimensional payoff type $\theta_i$, with virtual willingness to pay $\varphi_i^\theta(\theta_i)$. Attaching multiplier $\lambda$ to the posterior execution constraint, the bonus-expenditure term for bidder $i$ is
\[
\max_{0\leq q\leq \bar B_i x_i(\theta)} \left(\lambda\bar\Delta_i(\theta_i)-\bar\sigma_i(\theta_i)\right)q,
\]
which gives the shadow term $\bar\chi_i^\theta(\theta_i\mid\lambda)$. Thus, for a fixed multiplier, the allocation objective is
\[
\mathbb E\left[ \sum_{i=1}^n x_i(\theta) \bar S_i^\theta(\theta_i\mid\lambda) \right].
\]
The strict-feasibility condition \eqref{eq:payoff_relevant_private_slater} gives a supporting multiplier $\lambda^\theta$ by the same value function argument as Lemma~\ref{lem:multiplier}. Since $\bar S_i^\theta$ need not be increasing in $\theta_i$, the allocation is chosen from $\mathcal X^{IC,\theta}$. The envelope formula in willingness to pay implements the resulting monotone allocation. \qed
\bigskip

\noindent{\bf Proof of Proposition~\ref{prop:optimal_certification_given_signal}.}
Fix a certification technology $\mathcal C$. Since the signal profile $s$ is public before the auction, the designer can condition both the allocation rule and the execution contract on $s$. Conditional on $s$, the bidder side is a single-dimensional private value problem with conditional virtual values $\varphi_i(v_i\mid s_i)$. By the revenue identity, expected payments are bounded above by conditional expected virtual surplus, with equality for envelope payments and zero utility for the lowest type. Thus, conditional on inducing execution effort, the seller's expected virtual payoff is bounded by
\[
\E\left[ \sum_{i=1}^n x_i(v,s) \left( \varphi_i(v_i\mid s_i) + g_i^{\mathcal C}(v_i,s_i) \right) - \sum_{i=1}^n \sigma_i^{\mathcal C}(v_i,s_i)q_i(v,s) \right],
\]
subject to feasibility, monotonicity in values for each public signal profile, and the execution incentive constraint
\[
\E\left[ \sum_{i=1}^n \Delta_i^{\mathcal C}(v_i,s_i)q_i(v,s) \right]\geq \psi.
\]

Attach multiplier $\lambda\geq0$ to the execution incentive constraint. The Lagrangian is
\[
\E\left[ \sum_{i=1}^n x_i(v,s) \left( \varphi_i(v_i\mid s_i) + g_i^{\mathcal C}(v_i,s_i) \right) + \sum_{i=1}^n \left( \lambda\Delta_i^{\mathcal C}(v_i,s_i) - \sigma_i^{\mathcal C}(v_i,s_i) \right)q_i(v,s) \right] -\lambda\psi.
\]
For fixed $x_i(v,s)$, the optimal bonus expenditure solves
\[
\max_{0\leq q\leq B_i x_i(v,s)} \left( \lambda\Delta_i^{\mathcal C}(v_i,s_i) - \sigma_i^{\mathcal C}(v_i,s_i) \right)q.
\]
The maximized value per unit of allocation is
\[
B_i \max\left\{ \lambda\Delta_i^{\mathcal C}(v_i,s_i) - \sigma_i^{\mathcal C}(v_i,s_i), 0 \right\}.
\]
Substituting this value into the Lagrangian gives the certified shadow score
\[
S_i^{\mathcal C}(v_i,s_i\mid\lambda) = \varphi_i(v_i\mid s_i) + g_i^{\mathcal C}(v_i,s_i) + B_i \max\left\{ \lambda\Delta_i^{\mathcal C}(v_i,s_i) - \sigma_i^{\mathcal C}(v_i,s_i), 0 \right\}.
\]

The strict-feasibility condition \eqref{eq:certification_slater} gives a supporting multiplier $\lambda^{\mathcal C}\geq 0$ for the ex ante execution constraint. This multiplier is common across signal profiles. Substituting the optimized bonus expenditure into the Lagrangian, the seller chooses a signal-contingent allocation rule solving
\[
x^{\mathcal C} \in \arg\max_{x\in\mathcal X^{IC,\mathcal C}} \mathbb E \left[ \sum_{i=1}^n x_i(v,s)S_i^{\mathcal C} \left(v_i,s_i\mid\lambda^{\mathcal C}\right) \right].
\]
Because $s$ is public, this rule can be implemented conditional on each realized signal profile, with monotonicity imposed in values for that profile. The associated bonus expenditure is the maximizer stated in the proposition. Because $x^{\mathcal C}\in\mathcal X^{IC,\mathcal C}$, the allocation is monotone in each bidder's value for every public signal profile. The envelope payment formula conditional on $s$ implements the allocation in dominant strategies. \qed

\bigskip

\noindent{\bf Proof of Corollary~\ref{prop:value_of_certification}.}
If $\mathcal C'$ is more informative than $\mathcal C$, then the seller can garble or ignore the additional signal generated by $\mathcal C'$ and replicate the mechanism that is optimal under $\mathcal C$. Hence, $V(\mathcal C')\geq V(\mathcal C)$. The optimal certification choice follows by subtracting the certification cost $K(\mathcal C)$.\qed
\bigskip

\noindent{\bf Proof of Proposition~\ref{prop:optimal_externalities}.} By the revenue identity, expected payments are bounded above by expected virtual surplus. Conditional on inducing effort, the seller's expected virtual payoff, net of execution bonuses, is
\[
\E\left[ \sum_{i=1}^n x_i(v)\left(\varphi_i(v_i)+g_i(v)\right) - \sum_{i=1}^n \sigma_i(v)q_i(v) \right].
\]
The executor's incentive constraint is
\[
\E\left[ \sum_{i=1}^n \Delta_i(v)q_i(v) \right]\geq\psi.
\]
Attach multiplier $\lambda\geq0$. The Lagrangian is
\[
\E\left[ \sum_{i=1}^n x_i(v)\left(\varphi_i(v_i)+g_i(v)\right) + \sum_{i=1}^n \left(\lambda\Delta_i(v)-\sigma_i(v)\right)q_i(v) \right] -\lambda\psi.
\]
For fixed $x_i(v)$, the bonus expenditure term is maximized by solving
\[
\max_{0\leq q\leq B_i(v)x_i(v)} \left(\lambda\Delta_i(v)-\sigma_i(v)\right)q.
\]
The maximized value per unit of allocation is $B_i(v)\max\{\lambda\Delta_i(v)-\sigma_i(v),0\}$. Thus, the Lagrangian allocation objective is
\[
\E\left[ \sum_{i=1}^n x_i(v)\mathcal{E}_i(v\mid\lambda) \right].
\]
Because $\mathcal{E}_i(v\mid\lambda)$ may depend on other bidders' reports, pointwise maximization need not be monotone in each bidder's own report. The allocation must be chosen from $\mathcal X^{IC}$. The strict feasibility condition \eqref{eq:externality_slater} gives a supporting multiplier by the same value function argument as Lemma~\ref{lem:multiplier}. The envelope payment formula implements the resulting monotone allocation in dominant strategies. \qed
\bigskip

\noindent{\bf Proof of Lemma~\ref{prop:no_standard_scoring_representation}.} Suppose bidder $i$ strictly wins at report $v_i$ when the other reports are $v_{-i}$. Then, $s_i(v_i)>0$ and $s_i(v_i)>s_j(v_j)$ for all $j\neq i$. If $v_i'>v_i$, weak monotonicity gives $s_i(v_i')\geq s_i(v_i)$. Hence, $s_i(v_i')>0$ and $s_i(v_i')>s_j(v_j)$ for all $j\neq i$. As a result, bidder $i$ strictly wins at $v_i'$ as well.

In Example~\ref{ex:externality_nonmonotone}, fixing $v_1=0.51$, bidder 2 strictly wins when it reports $v_2=0.55$, but strictly loses when it reports $v_2=0.61$. This strict win-to-loss reversal cannot occur under any bidder-specific score representation with weakly increasing own scores. \qed

\bigskip

\noindent{\bf Proof of Proposition~\ref{prop:set_dependent_shadow_allocation}.} By the revenue identity,
\[
\E\left[\sum_i p_i(v)\right] \leq \E\left[\sum_i \varphi_i(v_i)x_i(v)\right].
\]
Since $x_i(v)=\sum_{A\ni i}y_A(v)$, expected virtual surplus can be written as
\[
\E\left[ \sum_{A\in\mathcal A_K} y_A(v)\sum_{i\in A}\varphi_i(v_i) \right].
\]
Thus, conditional on inducing effort, the seller maximizes
\[
\E\left[ \sum_{A\in\mathcal A_K} y_A(v) \left( \sum_{i\in A}\varphi_i(v_i)+g_A \right) - \sum_{A\in\mathcal A_K}\sigma_A q_A(v) \right]
\]
subject to feasibility and
\[
\E\left[ \sum_{A\in\mathcal A_K}\Delta_A q_A(v) \right]\geq\psi.
\]
Attach a multiplier $\lambda\geq0$. For fixed $y_A(v)$, the bonus expenditure for set $A$ solves
\[
\max_{0\leq q\leq B_A y_A(v)} (\lambda\Delta_A-\sigma_A)q,
\]
with maximized value $y_A(v)B_A\max\{\lambda\Delta_A-\sigma_A,0\}$. Hence, the Lagrangian objective for winner set $A$ is
\[
\sum_{i\in A}\varphi_i(v_i)+H_A(\lambda).
\]
Strict feasibility gives a supporting multiplier $\lambda^A$ by the same value function argument as Lemma~\ref{lem:multiplier}, and pointwise maximization gives the stated allocation.

It remains to show monotonicity. Fix bidder $i$ and $v_{-i}$. Increasing $v_i$ raises $\varphi_i(v_i)$ and thus it raises the objective of every winner set containing $i$ by the same amount, while leaving the objective of every winner set not containing $i$ unchanged. If bidder $i$ belongs to an optimal winner set at $v_i$, then at any higher report $v_i'>v_i$, some optimal winner set can be selected that also contains $i$. With a monotone tie-breaking rule, $x_i^\ast(v_i,v_{-i})$ is nondecreasing in $v_i$. The standard envelope payment formula implements the allocation in dominant strategies. \qed

\bigskip

\noindent{\bf Proof of Corollary~\ref{prop:nonmodular_no_individual_scoring}.} If $H_A(\lambda^A)$ is modular, then there exist numbers $\{\eta_i(\lambda^A)\}_{i\in N}$ such that
\[
H_A(\lambda^A)=\sum_{i\in A}\eta_i(\lambda^A)
\]
for every feasible winner set $A$. Hence,
\[
\sum_{i\in A}\varphi_i(v_i)+H_A(\lambda^A) = \sum_{i\in A}\bigl(\varphi_i(v_i)+\eta_i(\lambda^A)\bigr).
\]
The set-dependent rule is exactly a standard individual-score rule.

Conversely, suppose the set-dependent shadow objective can be written in standard individual-score form. Then, there must exist bidder-specific constants $\eta_i(\lambda^A)$ such that, for every feasible set $A$,
\[
\sum_{i\in A}\varphi_i(v_i)+H_A(\lambda^A) = \sum_{i\in A}\bigl(\varphi_i(v_i)+\eta_i(\lambda^A)\bigr),
\]
up to an additive constant common to all feasible sets. Since the virtual-surplus terms are the same on both sides, this requires
\[
H_A(\lambda^A)=\sum_{i\in A}\eta_i(\lambda^A)
\]
up to a common constant. Normalizing the outside option to zero, this is exactly modularity. Thus, when $H_A(\lambda^A)$ is not modular, the set-level execution adjustment cannot be reduced to bidder-specific score premia. \qed

\bigskip

\section{Extensions}

\subsection{Continuous effort}
\label{subsec:continuous_effort}

The binary effort specification is useful for transparency, but the shadow-score logic is not specific to binary effort. Suppose instead that the execution agent chooses effort $e\in[0,\bar e]$ before values are realized and before the winner is known. Effort cost is $C(e)$, where $C$ is continuously differentiable, strictly convex, satisfies $C(0)=0$, and $C'(e)>0$ for all $e>0$. If bidder $i$ wins, success occurs with probability
\[
\sigma_i(e)\coloneqq \sigma_i^0+e\D_i,
\]
where $\sigma_i(e)\in[0,1]$ for all $e\in[0,\bar e]$. Let $g_i(e)$ denote the seller's non-bonus execution payoff from assigning the object to bidder $i$ when the induced effort level is $e$.

Given an allocation-bonus pair $(x,q)$, define the execution incentive generated by the mechanism as
\[
I(x,q)\coloneqq \E\left[\sum_{i=1}^n \D_i q_i(v)\right].
\]
If the execution agent chooses effort $\hat e$, its expected payoff, up to terms independent of effort, is
\[
\hat e\, I(x,q)-C(\hat e).
\]
An interior target effort $e\in(0,\bar e)$ is induced whenever $I(x,q)=C'(e)$. Since $C$ is strictly convex, this condition is sufficient as well as necessary for $e$ to be the unique interior effort choice.

Fix an interior target effort $e\in(0,\bar e)$. Conditional on inducing this effort level, the seller's expected virtual payoff, net of execution bonuses, is
\[
\Pi_e(x,q) \coloneqq \E\left[ \sum_{i=1}^n x_i(v)\left(\varphi_i(v_i)+g_i(e)\right) - \sum_{i=1}^n \sigma_i(e)q_i(v) \right].
\]
The target effort $e$ is strictly feasible if
\begin{equation}
\exists\,(\tilde x,\tilde q) \quad\text{such that}\quad \sum_{i=1}^n \tilde x_i(v)\leq1,\quad 0\leq \tilde q_i(v)\leq B_i\tilde x_i(v), \quad\text{and}\quad I(\tilde x,\tilde q)>C'(e).
\label{eq:ce_strict_feasibility}
\end{equation}
For each bidder $i$, define
\[
\chi_i^e(\lambda^e) \coloneqq \max_{b\in[0,B_i]} \left(\lambda^e\D_i-\sigma_i(e)\right)b = B_i\max\{\lambda^e\D_i-\sigma_i(e),0\}.
\]
The associated continuous effort shadow score is
\[
S_i^e(v_i\mid \lambda^e) \coloneqq \varphi_i(v_i)+g_i(e)+\chi_i^e(\lambda^e).
\]

\begin{corollary}
\label{prop:continuous_effort_score}
Fix an interior target effort $e\in(0,\bar e)$ and suppose that \eqref{eq:ce_strict_feasibility} holds. Under Assumption~\ref{ass:regularity}, there exists an optimal mechanism inducing effort $e$ and a multiplier $\lambda^e\geq0$ such that the following statements hold.

The optimal allocation assigns the object only to bidders with maximal nonnegative continuous effort shadow score, i.e.,
\[
x_i^e(v)>0 \implies S_i^e(v_i\mid \lambda^e) = \max\left\{0,\max_{j\in N}S_j^e(v_j\mid \lambda^e)\right\}.
\]
The bonus expenditure satisfies
\[
q_i^e(v) \in \argmax_{0\leq q\leq B_i x_i^e(v)} \left(\lambda^e\D_i-\sigma_i(e)\right)q.
\]
The allocation is implementable in dominant strategies by a shadow-score auction with critical value payments.
\end{corollary}

\begin{proof}
For the fixed target effort $e$, consider the relaxed target problem
\[
\max_{x,q}\Pi_e(x,q)
\]
subject to
\[
\sum_{i=1}^n x_i(v)\leq1,
\]
\[
0\leq x_i(v)\leq1,
\]
\[
0\leq q_i(v)\leq B_i x_i(v),
\]
and
\[
I(x,q)\geq C'(e).
\]
Because bonus expenditures are costly and $\sigma_i(e)\geq0$, any feasible pair with $I(x,q)>C'(e)$ can be replaced by another feasible pair with the same allocation, weakly lower bonus expenditures, and $I(x,q)=C'(e)$. To see this, take
\[
\alpha=\frac{C'(e)}{I(x,q)}\in(0,1)
\]
and define $\bar q_i(v)=\alpha q_i(v)$. Then, $0\leq\bar q_i(v)\leq B_i x_i(v)$, $I(x,\bar q)=C'(e)$, and $\Pi_e(x,\bar q)\geq\Pi_e(x,q)$. It follows that the target problem has an optimal representative satisfying $I(x,q)=C'(e)$, which induces effort $e$.

By the strict feasibility condition \eqref{eq:ce_strict_feasibility}, the same value function argument used for the effort-inducing problem in the main text yields a supporting multiplier $\lambda^e\geq0$. The Lagrangian is
\[
\E\left[ \sum_{i=1}^n x_i(v)\left(\varphi_i(v_i)+g_i(e)\right) + \sum_{i=1}^n \left(\lambda^e\D_i-\sigma_i(e)\right)q_i(v) \right] -\lambda^e C'(e).
\]
For fixed $x_i(v)$, the bonus-expenditure component is maximized by solving
\[
\max_{0\leq q\leq B_i x_i(v)} \left(\lambda^e\D_i-\sigma_i(e)\right)q.
\]
The maximized value per unit of allocation is
\[
\chi_i^e(\lambda^e) = \max_{b\in[0,B_i]} \left(\lambda^e\D_i-\sigma_i(e)\right)b.
\]
Substituting this value into the Lagrangian gives the allocation objective
\[
\sum_{i=1}^n x_i(v) \left( \varphi_i(v_i)+g_i(e)+\chi_i^e(\lambda^e) \right).
\]
Thus, the Lagrangian is maximized pointwise by assigning the object to a bidder with the highest nonnegative continuous effort shadow score.

Finally, by Assumption~\ref{ass:regularity}, $\varphi_i(v_i)$ is strictly increasing in $v_i$. The remaining terms in $S_i^e(v_i\mid \lambda^e)$ do not depend on bidder $i$'s report. As a result, the allocation is monotone in each bidder's report and is implementable in dominant strategies by the corresponding critical value payment rule.
\end{proof}

Corollary~\ref{prop:continuous_effort_score} shows that the baseline shadow-score logic survives with continuous effort. Conditional on a target effort level, the optimal auction again ranks bidders by virtual values adjusted for the shadow value of relaxing the executor's incentive condition. The multiplier $\lambda^e$ is now the shadow value of the executor's first-order condition for the target effort level, rather than the shadow value of a binary effort constraint.

The seller can then choose among feasible target effort levels. Let $V^{\mathrm{ce}}(e)$ denote the value of the target problem characterized in Corollary~\ref{prop:continuous_effort_score}. The optimal continuous effort mechanism chooses an effort level
\[
e^\ast\in\argmax_{e\in[0,\bar e]} V^{\mathrm{ce}}(e),
\]
with the usual one-sided incentive conditions at the boundaries. At $e=0$, the condition is $I(x,q)\leq C'(0)$. When $C'(0)=0$, this is the no-effort case with no incentive bonuses. At $e=\bar e$, the condition is $I(x,q)\geq C'(\bar e)$. For every interior optimum, the allocation is a continuous effort shadow-score auction of the form characterized above.

\subsection{Private certification}
\label{app:private_certification}

The certification analysis in the main text assumes that the certified signal is public. This appendix shows that little changes when certification produces a signal observed by the seller but not disclosed to bidders before they bid, provided that the signal is relevant for execution but not informative about values. This case isolates the informational role of certification for execution incentives. If the seller privately observes a signal that is also informative about bidders' values, the standard revenue identity used below need not apply directly, because bidders' beliefs about the seller's signal may vary with their values.

This appendix assumes full commitment to the signal-contingent mechanism. The seller commits before observing the certified signal to the allocation rule, payment rule, and execution reward schedule that will be used after each signal realization. Equivalently, the signal can be interpreted as a hard certified record that can be used by the mechanism. If the seller instead privately observed a soft signal and could choose how to disclose or use it after bids were submitted, an additional incentive problem for the seller would arise. That information design problem is outside the scope of this appendix.

Let a certification technology $\mathcal C$ generate a signal $s_i\in S_i$ for each bidder $i$. The signal is observed by the seller before the auction, but is not observed by bidders when they submit their reports. Bidders know the certification technology and the mechanism. The seller commits to a mechanism that may condition allocations, bidder payments, and execution rewards on the realized signal profile $s=(s_1,\ldots,s_n)$. The signal is assumed to be independent of bidders' values. Thus, bidder values remain independently distributed according to the baseline distributions $F_i$, and the relevant virtual value is still
\[
\varphi_i(v_i) = v_i-\frac{1-F_i(v_i)}{f_i(v_i)}.
\]

The certified signal affects execution parameters. If bidder $i$ receives the object and the seller's certified signal is $s_i$, let $\Delta_i^{PC}(s_i)$, $\sigma_i^{PC}(s_i)$, and $g_i^{PC}(s_i)$ denote the marginal effect of execution effort, the success probability under effort, and the seller's non-bonus execution payoff. The bonus cap $B_i$ is public and does not depend on the private certified signal.

A signal-contingent direct mechanism consists of allocation rules $x_i(v,s)$, bidder payments $p_i(v,s)$, and promised bonus expenditures $q_i(v,s)$, with
\[
\sum_{i=1}^n x_i(v,s)\leq1,
\]
\[
0\leq x_i(v,s)\leq1,
\]
and
\[
0\leq q_i(v,s)\leq B_i x_i(v,s)
\]
for every $i$, $v$, and $s$. The execution incentive generated by $(x,q)$ is
\[
I^{PC}(x,q) \coloneqq  \E\left[ \sum_{i=1}^n \Delta_i^{PC}(s_i)q_i(v,s) \right],
\]
where the expectation is over values and certified signals. The seller private certification problem is strictly feasible if
\begin{equation}
\exists\,(\tilde x,\tilde q) \;\text{ such that }\; \sum_{i=1}^n \tilde x_i(v,s)\leq1,\quad 0\leq \tilde q_i(v,s)\leq B_i\tilde x_i(v,s), \;\text{ and }\; I^{PC}(\tilde x,\tilde q)>\psi.
\label{eq:private_certification_slater}
\end{equation}

For each $\lambda\geq0$, define the private certification shadow value of assigning the object to bidder $i$ when the seller observes signal $s_i$ as
\[
\chi_i^{PC}(s_i\mid\lambda) \coloneqq  B_i\max\{\lambda\Delta_i^{PC}(s_i)-\sigma_i^{PC}(s_i),0\}.
\]
The corresponding signal-contingent shadow score is
\[
S_i^{PC}(v_i,s_i\mid\lambda) \coloneqq  \varphi_i(v_i)+g_i^{PC}(s_i)+\chi_i^{PC}(s_i\mid\lambda).
\]
For each bidder $i$, define
\[
M_i^{PC}\left(v_{-i},s\mid\lambda\right) \coloneqq  \max\left\{ 0, \max_{j\neq i}S_j^{PC}(v_j,s_j\mid\lambda) \right\}.
\]

The next corollary shows that, under regularity and strict feasibility, the optimal effort-inducing mechanism is a signal-contingent shadow-score auction, with critical value payments computed conditional on the seller's certified signal.

\begin{corollary}
\label{prop:private_certification}
Suppose that Assumption~\ref{ass:regularity} and \eqref{eq:private_certification_slater} hold. Then, there exists an optimal effort-inducing mechanism and a multiplier $\lambda^{PC}\geq0$ such that the object is assigned only to bidders with maximal nonnegative private certification shadow score,
\[
x_i^{PC}(v,s)>0 \implies S_i^{PC}(v_i,s_i\mid\lambda^{PC}) = \max\left\{ 0, \max_{j\in N}S_j^{PC}(v_j,s_j\mid\lambda^{PC}) \right\}.
\]
The associated bonus expenditure satisfies
\[
q_i^{PC}(v,s) \in \argmax_{0\leq q\leq B_i x_i^{PC}(v,s)} \left( \lambda^{PC}\Delta_i^{PC}(s_i)-\sigma_i^{PC}(s_i) \right)q.
\]
The allocation is dominant-strategy implementable by a signal-contingent critical value payment rule. Bidder $i$ wins if and only if
\[
S_i^{PC}\left(v_i,s_i\mid\lambda^{PC}\right) \geq M_i^{PC}\left(v_{-i},s\mid\lambda^{PC}\right).
\]
When bidder $i$ wins, its payment is the critical value
\[
p_i(v_{-i},s) = \inf\left\{ z\in[\underline v_i,\overline v_i] \mid S_i^{PC}\left(z,s_i\mid\lambda^{PC}\right) \geq M_i^{PC}\left(v_{-i},s\mid\lambda^{PC}\right) \right\},
\]
where the critical value is kept within bidder $i$'s value support.
\end{corollary}

\begin{proof}
Because the certified signal is independent of bidder values, the bidder side of the problem is the baseline single-dimensional private value problem. The fact that the seller observes $s$, and bidders do not, does not change the virtual value. For any incentive compatible and individually rational mechanism,
\[
\E\left[ \sum_{i=1}^n p_i(v,s) \right] \leq \E\left[ \sum_{i=1}^n \varphi_i(v_i)x_i(v,s) \right],
\]
with equality under the envelope payment formula when the lowest type obtains zero utility.

Conditional on inducing execution effort, the seller's expected virtual payoff is bounded above by
\[
\E\left[ \sum_{i=1}^n x_i(v,s) \left( \varphi_i(v_i)+g_i^{PC}(s_i) \right) - \sum_{i=1}^n \sigma_i^{PC}(s_i)q_i(v,s) \right],
\]
subject to feasibility and the execution incentive constraint
\[
\E\left[ \sum_{i=1}^n \Delta_i^{PC}(s_i)q_i(v,s) \right]\geq\psi.
\]

Attach multiplier $\lambda\geq0$ to the execution incentive constraint. The Lagrangian is
\[
\E\left[ \sum_{i=1}^n x_i(v,s) \left( \varphi_i(v_i)+g_i^{PC}(s_i) \right) + \sum_{i=1}^n \left( \lambda\Delta_i^{PC}(s_i)-\sigma_i^{PC}(s_i) \right)q_i(v,s) \right] -\lambda\psi.
\]
For fixed $x_i(v,s)$, the bonus-expenditure term is maximized by solving
\[
\max_{0\leq q\leq B_i x_i(v,s)} \left( \lambda\Delta_i^{PC}(s_i)-\sigma_i^{PC}(s_i) \right)q.
\]
The maximized value per unit of allocation is
\[
B_i \max\left\{ \lambda\Delta_i^{PC}(s_i)-\sigma_i^{PC}(s_i),0 \right\} = \chi_i^{PC}(s_i\mid\lambda).
\]
Substituting this value into the Lagrangian gives the allocation objective
\[
\sum_{i=1}^n x_i(v,s)S_i^{PC}(v_i,s_i\mid\lambda).
\]
Thus, for a fixed multiplier, the Lagrangian is maximized pointwise by assigning the object to a bidder with the highest nonnegative private certification shadow score.

The strict-feasibility condition \eqref{eq:private_certification_slater} gives a supporting multiplier $\lambda^{PC}\geq0$ by the same value function argument as Lemma~\ref{lem:multiplier}. Hence, the pointwise maximizer at $\lambda^{PC}$, together with the stated bonus-expenditure rule, is optimal conditional on inducing execution effort.

It remains to verify implementability. For fixed $s$, $v_{-i}$, and $\lambda^{PC}$, bidder $i$'s score is
\[
S_i^{PC}\left(v_i,s_i\mid\lambda^{PC}\right) = \varphi_i(v_i) + g_i^{PC}(s_i) + \chi_i^{PC}\left(s_i\mid\lambda^{PC}\right).
\]
By Assumption~\ref{ass:regularity}, $\varphi_i(v_i)$ is strictly increasing in $v_i$, and all other terms are independent of bidder $i$'s report. Therefore, bidder $i$'s allocation is nondecreasing in its report for every realized signal profile $s$ and every $v_{-i}$. The signal-contingent critical value payment rule implements the allocation in dominant strategies. Since truthful reporting is optimal for every realized signal profile, it is also optimal even though bidders do not observe the signal before reporting.
\end{proof}

Corollary~\ref{prop:private_certification} shows that public disclosure of the certified signal is not essential when certification produces execution information rather than value information.\footnote{ Assumption~\ref{ass:regularity} is needed because the certified signal is private to the seller and the result uses a signal-contingent score rule. For each realized signal, monotonicity of $\varphi_i(v_i)$ ensures that bidder $i$'s score is increasing in its own report, so the allocation is implementable by signal-contingent critical values. } The seller can use the private certified signal to adjust scores and execution rewards, and the resulting allocation remains truthful because each signal-contingent score is increasing in the bidder's value. The difference from public certification is informational, as bidders may not know their realized score adjustment when they bid. The mechanism is nevertheless dominant-strategy implementable because the critical value rule is valid for every signal realization.

Let $V^{PC}(\mathcal C)$ denote the seller's gross value from the optimal effort-inducing mechanism under private certification technology $\mathcal C$, before subtracting certification costs. As before, say that $\mathcal C'$ is more informative than $\mathcal C$ if the signal generated by $\mathcal C$ can be obtained by garbling the signal generated by $\mathcal C'$. If $\mathcal C'$ is more informative than $\mathcal C$, then $V^{PC}(\mathcal C')\geq V^{PC}(\mathcal C)$. The seller chooses a private certification technology solving
\[
\max_{\mathcal C} \left\{ V^{PC}(\mathcal C)-K(\mathcal C) \right\}.
\]
If $\mathcal C'$ is more informative than $\mathcal C$, the seller can ignore the additional signal realizations generated by $\mathcal C'$ and replicate the mechanism that is optimal under $\mathcal C$. Therefore, the gross value under $\mathcal C'$ is at least the gross value under $\mathcal C$. The optimal certification choice follows by subtracting the ex ante cost of certification.


\addcontentsline{toc}{section}{References}
\bibliographystyle{aer}
\bibliography{auction_execution_refs}

@article{McAfeeMcMillan1986,
  author  = {McAfee, R. Preston and McMillan, John},
  title   = {Bidding for Contracts: A Principal-Agent Analysis},
  journal = {RAND Journal of Economics},
  year    = {1986},
  volume  = {17},
  number  = {3},
  pages   = {326--338}
}

@article{LaffontTirole1987,
  author  = {Laffont, Jean-Jacques and Tirole, Jean},
  title   = {Auctioning Incentive Contracts},
  journal = {Journal of Political Economy},
  year    = {1987},
  volume  = {95},
  number  = {5},
  pages   = {921--937},
  doi     = {10.1086/261496}
}

@article{Branco1997,
  author  = {Branco, Fernando},
  title   = {The Design of Multidimensional Auctions},
  journal = {RAND Journal of Economics},
  year    = {1997},
  volume  = {28},
  number  = {1},
  pages   = {63--81},
  doi     = {10.2307/2555940}
}

@article{ArozamenaCantillon2004,
  author  = {Arozamena, Leandro and Cantillon, Estelle},
  title   = {Investment Incentives in Procurement Auctions},
  journal = {Review of Economic Studies},
  year    = {2004},
  volume  = {71},
  number  = {1},
  pages   = {1--18},
  doi     = {10.1111/0034-6527.00273}
}

@article{EsoSzentes2007,
  author  = {Es{\H{o}}, P{\'e}ter and Szentes, Bal{\'a}zs},
  title   = {Optimal Information Disclosure in Auctions and the Handicap Auction},
  journal = {Review of Economic Studies},
  year    = {2007},
  volume  = {74},
  number  = {3},
  pages   = {705--731},
  doi     = {10.1111/j.1467-937X.2007.00442.x}
}

@article{ChakrabortyKhalilLawarree2021,
  author  = {Chakraborty, Indranil and Khalil, Fahad and Lawarr{\'e}e, Jacques},
  title   = {Competitive Procurement with Ex Post Moral Hazard},
  journal = {RAND Journal of Economics},
  year    = {2021},
  volume  = {52},
  number  = {1},
  pages   = {179--206},
  doi     = {10.1111/1756-2171.12366}
}

@article{GershkovMoldovanuStrackZhang2021,
  author  = {Gershkov, Alex and Moldovanu, Benny and Strack, Philipp and Zhang, Mengxi},
  title   = {A Theory of Auctions with Endogenous Valuations},
  journal = {Journal of Political Economy},
  year    = {2021},
  volume  = {129},
  number  = {4},
  pages   = {1011--1051},
  doi     = {10.1086/712735}
}

@article{AskerCantillon2008,
  author  = {Asker, John and Cantillon, Estelle},
  title   = {Properties of Scoring Auctions},
  journal = {RAND Journal of Economics},
  year    = {2008},
  volume  = {39},
  number  = {1},
  pages   = {69--85}
}

@article{AskerCantillon2010,
  author  = {Asker, John and Cantillon, Estelle},
  title   = {Procurement When Price and Quality Matter},
  journal = {RAND Journal of Economics},
  year    = {2010},
  volume  = {41},
  number  = {1},
  pages   = {1--34}
}

@article{Che1993,
  author  = {Che, Yeon-Koo},
  title   = {Design Competition through Multidimensional Auctions},
  journal = {RAND Journal of Economics},
  year    = {1993},
  volume  = {24},
  number  = {4},
  pages   = {668--680}
}

@article{Holmstrom1979,
  author  = {Holmstr{\"o}m, Bengt},
  title   = {Moral Hazard and Observability},
  journal = {Bell Journal of Economics},
  year    = {1979},
  volume  = {10},
  number  = {1},
  pages   = {74--91}
}

@article{Innes1990,
  author  = {Innes, Robert D.},
  title   = {Limited Liability and Incentive Contracting with Ex-Ante Action Choices},
  journal = {Journal of Economic Theory},
  year    = {1990},
  volume  = {52},
  number  = {1},
  pages   = {45--67}
}

@article{JehielMoldovanuStacchetti1996,
  author  = {Jehiel, Philippe and Moldovanu, Benny and Stacchetti, Ennio},
  title   = {How (Not) to Sell Nuclear Weapons},
  journal = {American Economic Review},
  year    = {1996},
  volume  = {86},
  number  = {4},
  pages   = {814--829}
}

@article{JehielMoldovanuStacchetti1999,
  author  = {Jehiel, Philippe and Moldovanu, Benny and Stacchetti, Ennio},
  title   = {Multidimensional Mechanism Design for Auctions with Externalities},
  journal = {Journal of Economic Theory},
  year    = {1999},
  volume  = {85},
  number  = {2},
  pages   = {258--293}
}

@article{Myerson1981,
  author  = {Myerson, Roger B.},
  title   = {Optimal Auction Design},
  journal = {Mathematics of Operations Research},
  year    = {1981},
  volume  = {6},
  number  = {1},
  pages   = {58--73},
  doi     = {10.1287/moor.6.1.58}
}

@article{RileySamuelson1981,
  author  = {Riley, John G. and Samuelson, William F.},
  title   = {Optimal Auctions},
  journal = {American Economic Review},
  year    = {1981},
  volume  = {71},
  number  = {3},
  pages   = {381--392}
}

@book{Sundaram1996,
  author    = {Sundaram, Rangarajan K.},
  title     = {A First Course in Optimization Theory},
  publisher = {Cambridge University Press},
  year      = {1996}
}

@book{LaffontTirole1993,
  author    = {Laffont, Jean-Jacques and Tirole, Jean},
  title     = {A Theory of Incentives in Procurement and Regulation},
  publisher = {MIT Press},
  year      = {1993},
  address   = {Cambridge, MA}
}

@misc{FHWA_AB_Bidding,
  author       = {{Federal Highway Administration}},
  title        = {Cost-plus-Time (A+B) Bidding},
  year         = {2023},
  howpublished = {Construction Program Guide, U.S. Department of Transportation},
  note         = {Updated November 7, 2023}
}

@misc{FAR_15305,
  author       = {{Federal Acquisition Regulation}},
  title        = {FAR 15.305: Proposal Evaluation},
  year         = {2026},
  howpublished = {Acquisition.gov}
}

@misc{FAR_421501,
  author       = {{Federal Acquisition Regulation}},
  title        = {FAR 42.1501: Contractor Performance Information},
  year         = {2026},
  howpublished = {Acquisition.gov}
}


\end{document}